\documentclass[11pt]{article} 

\usepackage[a4paper,left=2.5cm,right=2.5cm,top=2.5cm,bottom=3.0cm]{geometry}
\usepackage[dvips]{epsfig}
\usepackage{amsfonts}
\usepackage{moreverb}
\usepackage[normalem]{ulem}

\usepackage{amsthm}

\newtheorem{theorem}{Theorem}[section]
\newtheorem{lemma}{Lemma}[section]
\newtheorem{corollary}{Corollary}[section]

\newtheorem{definition}{Definition}[section]
\newtheorem{observation}{Observation}[section]

\usepackage{graphicx}
\usepackage{xcolor}
\usepackage{enumitem}
\usepackage{amssymb}
\usepackage{amsmath}

\usepackage[noend, noline, ruled, linesnumbered]{algorithm2e} %allowed
\SetAlFnt{\normalsize}
\DontPrintSemicolon
\SetKwComment{Comment}{$\triangleright$\ }{}
\SetKwProg{Def}{def}{:}{}
\SetArgSty{textnormal}
\SetKwRepeat{Do}{do}{while}
\SetKwComment{Comment}{$\triangleright$\ }{}

\newcommand{\btd}[2]{btd(#1, #2)}
\newcommand{\Lev}[2]{L_{#1}(#2)}
\newcommand{\lev}[2]{\ell_{#1}(#2)}
\newcommand{\ET}[1]{\mathcal{E}_{#1}}

\newcommand{\dg}[2]{d_{#1}(#2)}
\newcommand{\Fvis}[2]{\mathrm{vis}(#1,#2)}  
\newcommand{\BTD}{{\textsc{BTD}}}

\newcommand{\MHS}{{\textsc{MHS}}}
\newcommand{\F}{{F}}
\newcommand{\B}{\mathit{b}}
\newcommand{\M}{\mathit{M}}
\newcommand{\NP}{\mathcal{NP}}
\newcommand{\mcP}{\mathcal{P}}   

\newcommand{\rbar}{\overline{r}}
\newcommand{\mbar}{\overline{m}}
\newcommand{\Tbar}{\overline{T}}
\newcommand{\Tbard}{\overline{T}_{d}}
\newcommand{\PM}{{P_{2^\M}}}

\newcommand{\Gb}{G}
\newcommand{\Gbi}{{\Gb(a_i)}}
\newcommand{\ka}{{\alpha}}
\newcommand{\Gk}{G_{\ka}}

\newcommand{\A}{A}

\begin{document}

\title{\bf The Complexity of Bicriteria Tree-Depth}

\author{
Piotr Borowiecki\footnotemark[1]
\and
Dariusz Dereniowski\footnotemark[1] \footnotemark[2]
\and
Dorota Osula\footnotemark[1] \footnotemark[2]
}

\maketitle
\def\thefootnote{\fnsymbol{footnote}}

\footnotetext[1]{Faculty of Electronics, Telecommunications and Informatics, Gda\'{n}sk~University~of~Technology, Poland. Work partially supported under Ministry of Science and Higher Education (Poland) subsidy for Gda\'{n}sk University of Technology.}

\footnotetext[2]{Authors partially supported by National Science Centre, Poland, grant number 2018/31/B/ST6/00820.}

\maketitle

\begin{abstract}
The tree-depth problem can be seen as finding an elimination tree of minimum height for a given input graph $G$. We introduce a bicriteria generalization in which additionally the \emph{width} of the elimination tree needs to be bounded by some input integer $b$. We are interested in the case when $G$ is the line graph of a tree, proving that the problem is $\NP$-hard and obtaining a polynomial-time additive $2b$-approximation algorithm. This particular class of graphs received significant attention in the past, mainly due to a number of potential applications, e.g. in parallel assembly of modular products, or parallel query processing in relational databases, as well as purely combinatorial applications, including searching in tree-like partial orders (which in turn generalizes binary search on sorted data).
\end{abstract}


\section{Introduction}

The problem of computing tree-depth has a long history in the realm of parallel computations as it emerged from different applications and under different names.
It was considered for the first time under the name of \emph{minimum height elimination trees} where it played an important role in parallel factorization of (sparse) matrices \cite{Liu90}.
Although the sparsity assumption is not necessary from the graph-theoretic standpoint in this application, this is the scenario where using the parallel (or distributed) approach to factorization provides a speed-up.

Then, the problem re-appeared under the name of \emph{vertex ranking} \cite{BodlaenderDJKKMT98}.
More applications have been brought up, including parallel assembly of multi-part products, where tree-like structures have been mostly considered.
In this particular scenario, an extension of this problem called \emph{$c$-edge ranking} \cite{ZhouKN96} has been introduced to model two sources of parallelism --- both the one that is captured by the tree-depth and the one that can be simulated by more powerful multi-arm robots, i.e., the robots that may perform $c$ operations at once as opposed to doing one at a time.
Another application in this realm includes parallel query processing in relational databases \cite{MakinoUI01}.
Here, the graph $G$ constructed on the basis of a database query does not need to be a tree but it turns out that in order to design an efficient parallel schedule for performing the query, a vertex ranking of a line graph of a certain spanning tree of $G$ is computed along the way.
As has been pointed out in \cite{MakinoUI01}, this scenario can be seen as a special case of a more general process of information retrieval in an arbitrary network $G$: suppose that each node initially holds a piece of information and one wants to merge/collect all the data at one site, i.e., one node.
Similarly as in database query processing above, where a single query could not participate in two different join operations simultaneously, we assume in this scheme that the data of a particular node cannot participate in more than one merging process at the same time.
This may be e.g. due to integrity reasons.
On the other hand, two data merging processes can occur in parallel as long as they comprise of pairwise disjoint node sets.
Under such assumptions, the minimum number of parallel steps (as well as the corresponding merging schedule) required to gather all data at one node can be computed by finding an optimal vertex ranking of a line graph of a certain spanning tree of $G$.
The above applications, by nature, primarily stimulate research on algorithms for line graphs of trees.
In fact, the typical approach was either to work, i.e., compute its tree-depth, on a line graph of a tree that models the parallel process, or to work on a spanning tree of the input graph that models the process.

Later on, the same problem has been introduced under different names in a number of other applications: it has been called \emph{LIFO-search} \cite{GiannopoulouHT12} in the area of graph pursuit-evasion games; it is a problem of searching in tree-like partial orders, which naturally generalizes the classical binary search problem on sorted arrays; it has been also named as \emph{ordered coloring} and more recently as \emph{tree-depth} \cite{NesetrilM06}.
Through the connection to searching partial orders, it is worth pointing out that computing tree-depth (or equivalently, a search strategy for the corresponding partial order) can be used to finding, in an automated way, software bugs \cite{Ben-AsherFN99}.
In this application, again, the line graphs of trees and their tree-depth are of interest.
To intuitively see the link between tree-depth and parallel processing, we may see the graph-theoretic problem from two perspectives.
One perspective that is the driving approach e.g. in parallel matrix factorization or in software testing, looks at the tree-depth in a `top-down' fashion: find a good separator in the graph, process the workload that corresponds to the separator and then recurse (in parallel) on the connected components.
What is powerful regarding the tree-depth problem is that it balances the ratio between the sizes of such separators and the recursion depth, providing the minimum overall computation time.
It is worth noting that neither focusing on small separators (and hence pay with larger recursion depth) nor finding the computation scheme with the minimum recursion depth (but perhaps using some larger separators) is an efficient approach for such applications.
The other perspective is to see the tree-depth of a line graph of a tree in a `bottom-up' fashion: find a matching in a tree, contract its edges in parallel, and repeat this step until you reach a state where there is a single vertex left.
This interpretation lies behind the above-mentioned applications in databases, parallel assemblies, and information retrieval (see Figure~\ref{fig:parallel} for an illustration of this scheme).
   \begin{figure}[htb!]
    \begin{center}
    \includegraphics[scale=0.8]{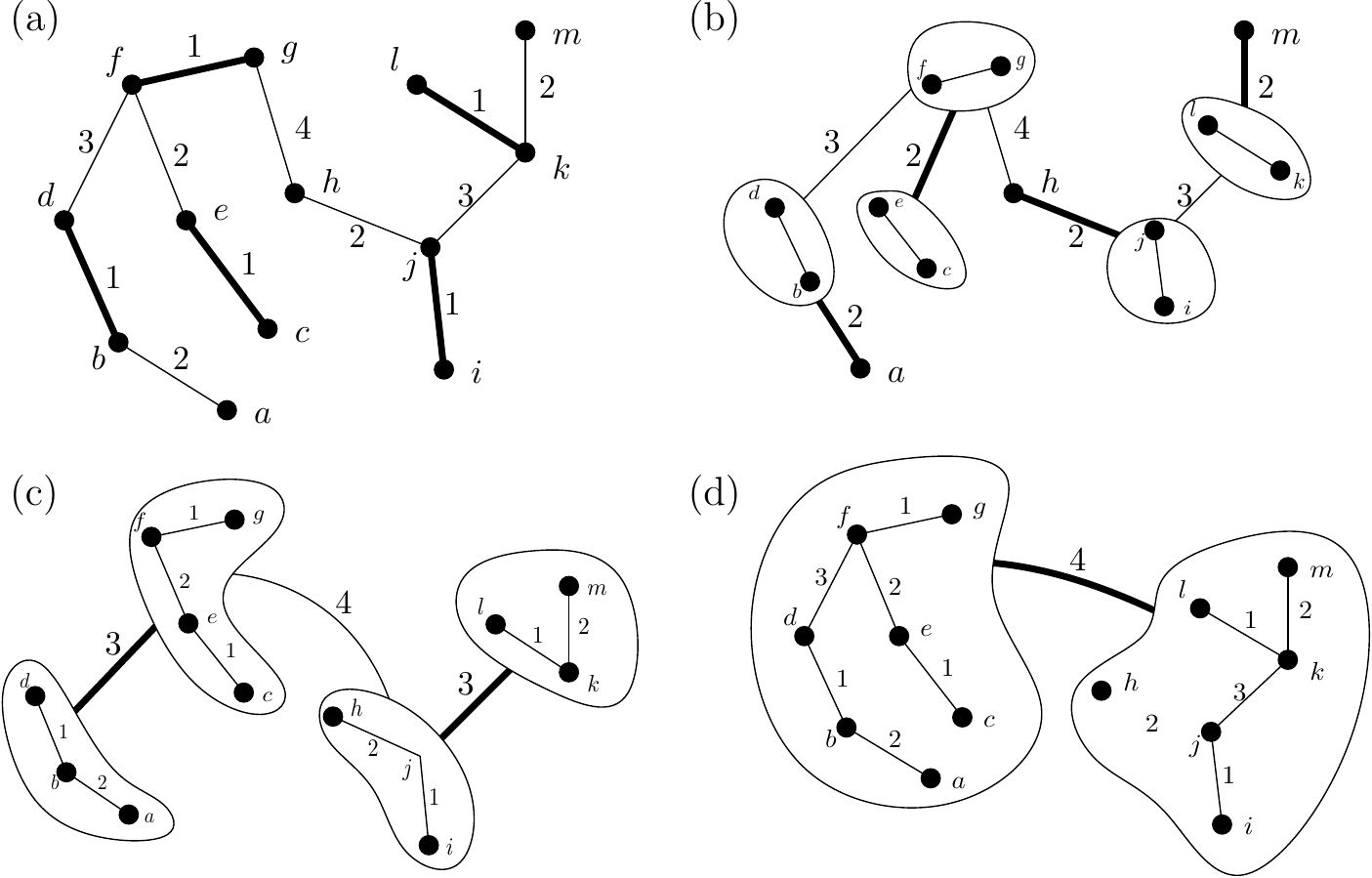}
    \caption{The process of parallel merging, in each step of which one is allowed to select a matching and contract its edges. See (a) for a tree with labels denoting consecutive matchings (the edges labeled $i$ are contracted in step $i$), and (b)-(d) for the results after the first three steps.}
    \label{fig:parallel}
   \end{center}
    \end{figure}
Intuitively, these two processes are two sides of the same coin: the edges of the matching selected in the $i$-th step correspond to the $i$-th lowest level in the recursion tree.

% .................................................................
\subsection{Related work}\label{sec:relwork}

From the complexity standpoint, we remark that the tree-depth problem is $\NP$-complete for arbitrary line graphs \cite{LamY98}.
Since in our work we focus on trees and their line graphs, we review here mostly works that relate to this class of graphs.
The minimum known superclass of graphs that generalizes trees and for which the problem is known to be $\NP$-complete is the family of chordal graphs \cite{DereniowskiN06}.
For the trees themselves, there has been a few papers dealing with computing optimal tree-depth, concluding in a linear-time algorithm in \cite{Schaffer89}.
The problem turned out to be much more challenging and interesting for line graphs of trees (a proper subclass of block graphs).
A number of papers have been published, see e.g. \cite{ZhouKN96,ZhouN94,ZhouN95}, that were gradually reducing the complexity from $O(n^4\log^3n)$ in \cite{Ben-AsherFN99} and $O(n^3\log n)$ in \cite{TorreGS95} to the final linear-time algorithms in~\cite{LamY01,MozesOW08}, where $n$ is the order of the input tree.\footnote{We remark that this research has been independently progressing as a tree-like partial order search problem and as an edge ranking problem until it was observed in~\cite{Dereniowski08} that the two are equivalent.}

Mostly motivated by the applications, there are some variants of the tree-depth problem of line graphs of trees introduced in the literature.
One of them considers weighted trees \cite{Dereniowski06} (equivalent to multitrees when weights are integral).
The weights naturally represent the number of time units required for performing particular operation (in the earlier mentioned applications, we took a silent assumption that all operations have unit duration).
Once we introduce the weights, the problem becomes strongly $\NP$-hard even for some restricted classes of trees \cite{CicaleseJLV12,CicaleseKLPV16,Dereniowski06}.
For algorithmic results on weighted paths see e.g. \cite{CicaleseJLV12,LaberMP02}.
Note that even such a simple structure as paths is practically interesting in this context since our problem generalizes binary search.
For weighted trees, there is a number of works improving on possible approximation ratio achievable in polynomial-time \cite{CicaleseJLV12,CicaleseKLPV16,Dereniowski06}, with the best one to date having an approximation factor of $O(\sqrt{\log n})$ \cite{DereniowskiKUZ17}, and it is unknown whether this is best possible---in particular, a challenging open question regarding this line of research is whether a constant-factor approximation is feasible for weighted trees. 
Also, many intriguing problems remain open in the area of online competitiveness that has been considered, e.g., for line graphs of specific star-like graphs \cite{bode13} or just for vertex variant on trees \cite{McDonald15}.

Another variant of the problem is the $c$-edge ranking mentioned above, for which an optimal solution can be found efficiently for an unweighted line graph of a tree \cite{ZhouKN96}.
Yet another relaxation of vertex ranking, studied for several classes of graphs, including trees, has been recently introduced in~\cite{KarpasNS15}.

We remark that there is a number of searching models that generalize searching in (rooted) tree-like partial orders to more general classes of partial orders see e.g. \cite{CarmoDKL04,Dereniowski08,OnakP06}.
Yet even more general query models have been introduced recently for graphs \cite{Emamjomeh-Zadeh16}.
These generalizations lose their connection to the tree-depth notion.
Interestingly, these general graph-theoretic models find applications in some machine learning algorithms~\cite{Emamjomeh-Zadeh17}.

We finish by mentioning that the tree-depth generalization that we consider in this work for line graphs of trees has been introduced and studied for trees in \cite{Zwaan10} under the name of \emph{vertex ranking with capacity}.
It has been shown in \cite{Zwaan10} that there exists a $O^*(2.5875^n)$-time optimal algorithm for general graphs, and that an $f(n)$-approximate solution to vertex ranking can be transformed to an $(f(n)+1)$-approximate solution to vertex ranking with capacity.
For trees, this problem admits an absolute $O(\log \B)$-approximation in polynomial time \cite{Zwaan10}.

% .................................................................
\subsection{Problem statement and our results}\label{sec:problem}

In this section we introduce a generalization of the concept of the classical elimination tree by considering elimination forests with level functions explicitly defined on their vertex sets. 
In this context a \emph{rooted forest} is meant as a disjoint union of rooted trees.
We start with a notion of the classical elimination tree that we call here a free elimination tree.
For the sake of correctness we point out that all graphs $G=(V,E)$ considered in this paper are finite, simple and undirected, with vertex set $V$ and edge set $E$.

\begin{definition}\label{def:freeET}
A \emph{free elimination tree} for a connected graph $G$ is a rooted tree $T$ defined recursively as follows:
\begin{enumerate}
    \item let $V(T)=V(G)$ and let an arbitrary vertex $r\in V(T)$ be the root of $T$,
    \item if $\vert V(G)\vert = 1$, then let $E(T) = \emptyset$.
          Otherwise, let 
            \[
               E(T) = \bigcup\limits_{i=1}^{k} \Bigl( E(T_i) \cup \{e_i\} \Bigr)
            \]
          where $k$ is the number of connected components of $G-r$, and $T_i$ stands for an elimination tree for the $i$-th connected component of $G-r$ with the root $r(T_i)$ joined by the edge $e_i$ with the root $r$ of $T$, i.e., $e_i=\{r(T_i),r\}$.
\end{enumerate}
\end{definition}

\begin{definition}
 Let $G$ be a graph with $k$ connected components. A \emph{free elimination forest} for $G$ is the disjoint union of $k$ free elimination trees, each of which is determined for distinct connected component of $G$. 
\end{definition}

Before the definition of elimination forest we point out that the notions of ancestor, parent and child are used in their usual sense. Namely, given two vertices $u$ and $v$ of a rooted forest $F$, we say that $v$ is \emph{an ancestor} of $u$ if $v$ belongs to the path with the end-vertices in $u$ and the root of the connected component of $F$ that contains $u$. If $v$ is an ancestor of $u$ and $\{v,u\}\in E(F)$, then $v$ is the \emph{parent} of $u$ while $u$ is a \emph{child} of $v$.

\begin{definition}
   Let $F$ be a free elimination forest for a graph $G$, and let $f: V(F)\rightarrow\mathbb{Z}^{+}$ be a \emph{level function}, i.e. a function such that $f(u)<f(v)$ whenever $v$ is an ancestor of $u$. A free elimination forest $F$ with a level function $f$ is called an \emph{elimination forest} for $G$ and it is denoted by $F_f$.
\end{definition}

For an elimination forest $F_f$ its \emph{height} $h(F_f)$ is defined as $\max\{f(v)\;\vert\;v\in V(F_f)\}$. Clearly, the maximum can be attained only for the roots of the connected components of $F_f$.
Now, for every $i\in\{1,\ldots,h(F_f)\}$ we define the \emph{$i$-th level} $\Lev{i}{F_f}$ of $F_f$ as the set of those vertices $v$ in $V(F_f)$ for which $f(v)=i$ (notice that the definition of a level function allows empty levels).
In a natural way the \emph{width $w(F_f)$} of elimination forest $F_f$ is defined as $\max_{i}\vert\Lev{i}{F_f}\vert$, where $i\in\{1,\ldots,h(F_f)\}$.

We point out that the above definitions do not impose the placement of the roots of all $k$ connected components of $F_f$ at the highest level. In fact, each root can be placed at an arbitrary level.
Moreover, it is also not required that adjacent vertices of an elimination forest occupy consecutive levels---there may be a gap of arbitrary size between the numbers of levels they belong to.

It is also worth mentioning that though the definition of the classical elimination tree (recall that we call it free elimination tree) does not explicitly give any level function, one of the possible functions can be, and usually is, implicitly deduced by assuming that each level is formed by a single recursive step in Definition \ref{def:freeET}.

\begin{definition}
 Let $G$ be a graph and let $\B$ be a positive integer. The \emph{bounded-width tree-depth $\btd{G}{\B}$} of a graph $G$ is the minimum $k$ for which there exists an elimination forest $F_f$ for a graph $G$ such that $h(F_f)=k$ and $w(F_f)\le \B$.
\end{definition}

Note that for every $G$ and $\B>0$ there always exists some elimination forest of width bounded by $\B$. Since in what follows the level function $f$ is always clear from the context, we omit $f$ in the symbol of an elimination forest.

We can now formulate our main problem.

\medskip
      \indent\textsc{Bounded-Width Tree-Depth ($\BTD$)}\\
      \indent\quad\textit{Input:}  A graph $G$, positive integers $k$ and $\B$.\\
      \indent\quad\textit{Question:} Does $\btd{G}{\B}\le k$ hold?

\medskip
The above seemingly small differences in the classical and our definitions (e.g. dropping the assumption on connectivity of a graph $G$, relaxing the requirements on level function) play an important role both in our $\NP$-completeness reduction and in our algorithm. They significantly affect the bounded-width tree-depth problem complexity thus making it different than that of the classical tree-depth.
The classical tree-depth problem (a variant without the bound $\B$) can be solved in linear time for line graphs of trees \cite{LamY01,MozesOW08}. The generalization we consider turns out to be $\NP$-complete.

\begin{theorem} \label{thm:NPC_trees1}
    $\BTD$ problem is $\NP$-complete for line graphs of trees.
\end{theorem}

The proofs given in Section~\ref{sec:NPC} reveal that bounded-width tree-depth behaves differently than the original tree-depth problem strongly depending on the connectivity of the graph.
Specifically, for tree-depth, only connected graphs are of interest since tree-depth of a non-connected graph is just the maximum tree-depth taken over its connected components.
In the bounded-width tree-depth problem, the connected components interact. En route of proving Theorem~\ref{thm:NPC_trees1}, we first obtain hardness of the problem for line graphs of forests and then we extend it to get the $\NP$-completeness of $\BTD$ for line graphs of trees.

On the positive side, we develop an approximation algorithm for line graphs of trees.
Here, the fact that minimum height elimination tree (or equivalently an optimal tree-depth) can be found efficiently \cite{LamY01,MozesOW08} turns out to be very useful---our approach is to start with such a tree (with some additional preprocessing) and squeeze it down so that its width becomes as required.
The squeezing-down is done in a top-down fashion, by considering the highest level that is too wide, i.e., contains more than $\B$ vertices, and moving the excess vertices downwards. Note that the algorithm needs to choose which vertices should go down---we move those vertices that are the roots of the subtrees of the smallest height.
This leads to polynomial-time approximation algorithm with an additive error of $2\B$.

\begin{theorem} \label{thm:ALG}
There exists a polynomial-time additive $2\B$-approximation algorithm for $\BTD$ problem for line graphs of trees. 
\end{theorem}

The worst-case guarantee of our algorithm is similar to the one presented in \cite{Zwaan10} for trees.
However, we point out two differences: finding balanced vertex separators for trees can be done easily and this is the foundation of the algorithm in \cite{Zwaan10}.
In contrast, for line graphs of trees such balanced separators do not exist in general, which is the reason for initializing the algorithm using a minimum height elimination tree. This step is important, for otherwise applying very involved bottom-up dynamic programming ideas from \cite{LamY01,MozesOW08} seems unavoidable.
The second difference lies in the fact that it remains unknown if there exists an optimal polynomial-time algorithm for trees, as no hardness argument is known.

% =================================================================
\section{Preliminaries} \label{sec:Preliminaries}

For the sake of clarity and to avoid involved notation and argument, necessary when carrying the proofs directly on line graphs, we use  $\ET{G}$ to denote an elimination forest for the line graph $L(G)$ of a graph $G$.
Consequently, since by the definitions of line graph and elimination forest, there is a natural one-to-one correspondence between the edges of a graph $G$ and the vertices of its line graph and hence the vertices of an elimination forest $\ET{G}$, for \emph{each edge} of $G$ we shortly say that it \emph{corresponds} to the appropriate vertex $v$ of $\ET{G}$ and that it \emph{belongs} to the level of $\ET{G}$ that contains $v$.
Notice that with a small abuse of notation the above one-to-one correspondence allows the use of any level function $\ell$, determined for an elimination forest $\ET{G}$, as if it was defined on $E(G)$. 
Thus for every edge $e$ in $E(G)$ we say that $\ell(e)=p$ if for the corresponding vertex $v$ of $\ET{G}$ it holds $v\in\Lev{p}{\ET{G}}$.

 \begin{figure}[htb!]
    \begin{center}
    %\vspace{4cm}
    \includegraphics[scale=0.9]{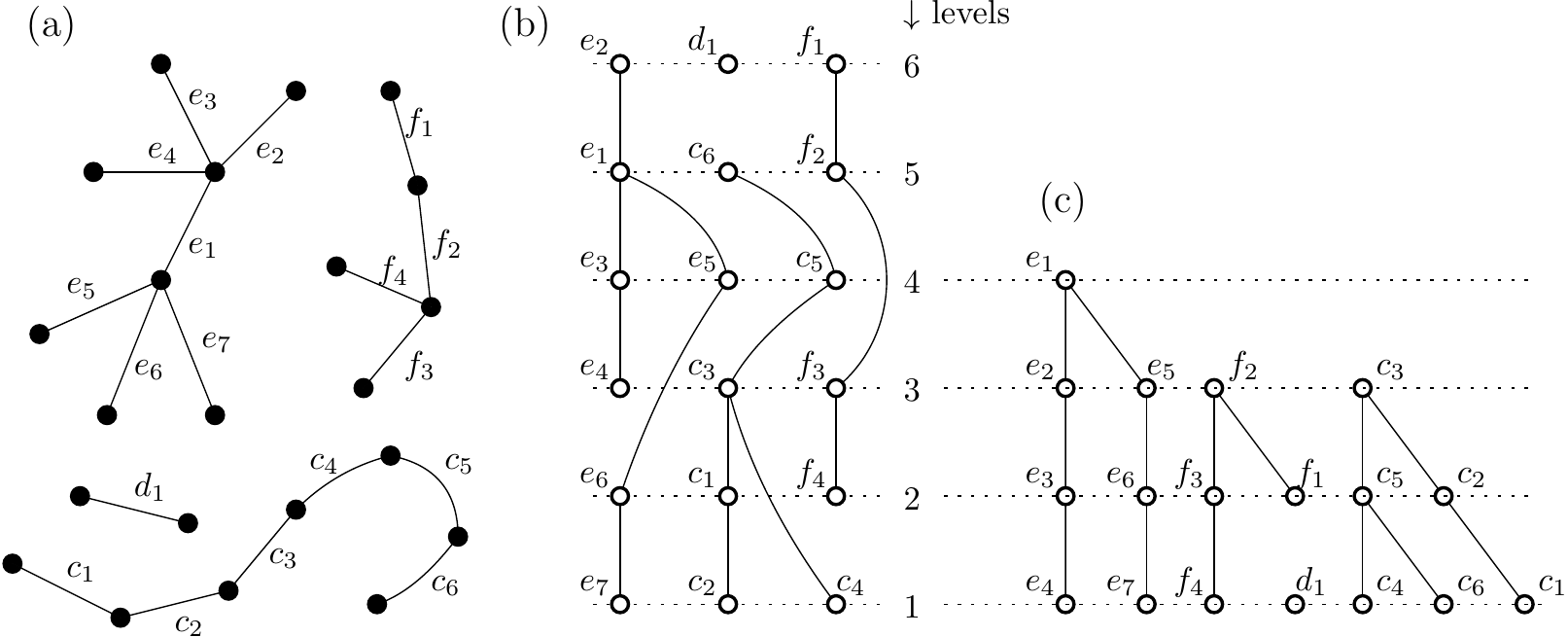}
    \caption{
       (a) a graph $G$ with 4 connected components; 
       (b) a minimum height elimination forest $\ET{G}$ when $\B=3$;
       (c) a minimum height elimination forest $\ET{G}$ when $\B\ge 7$ (here, an elimination tree for \emph{each} connected component of $G$ has minimum height).
       }
    \label{fig:elimtrees}
   \end{center}
    \end{figure}

Concerning the properties of level functions in the above-mentioned common context of a graph $G$ and elimination forest $\ET{G}$, we note that if two distinct edges $e_1, e_2$ are adjacent in $G$, then they cannot belong to the same level of $\ET{G}$, i.e., $\lev{}{e_1}\neq\lev{}{e_2}$.
Similarly, it is not hard to see that in the recursive elimination process performed on the vertices of $L(G)$ (according to Definition \ref{def:freeET}) all vertices eliminated at the same recursive step belong to distinct components of the processed graph, and for each of them the value of a level function is greater than for the vertices eliminated in further steps. 
In other words, for distinct edges $e_1, e_2$ of $G$ such that $\lev{}{e_1}=\lev{}{e_2}$ every path with end-edges $e_1, e_2$ contains an edge $e'$ such that $\lev{}{e'}>\lev{}{e_1}$.
It is worth mentioning that if we consider $\ell$ as a function defined on $E(G)$ and satisfying both of the above-mentioned properties, then $\ell$ can be equivalently seen as an edge ranking of a graph $G$ (see e.g. \cite{IyerRV91}). 
An example in Figure \ref{fig:elimtrees} serves as an illustration of the above concepts as well as those introduced in Section \ref{sec:problem}.

We use the same common context to define the visibility of a level from a vertex in $G$. Namely, for a vertex $v$ in $G$ we say that the $p$-th level is \emph{visible in $G$ from $v$} if there \emph{exists} an edge $e=\{u_1,u_2\}$ with $\lev{}{e}=p$ and $G$ contains a path with end-vertices $v,u$, where $u\in\{u_1,u_2\}$ and for each edge $e'$ of the path $\lev{}{e'}\le p$.
The set of levels visible in $G$ from $v$ is denoted by $\Fvis{G}{v}$.
When determining the levels admissible for a given edge  in a graph $G$ (or for the corresponding vertex in elimination forest $\ET{G}$) we need to consider and forbid all levels visible from both end-vertices of that edge, i.e., the level $p$ is \emph{admissible} for $e=\{u_1,u_2\}$ if neither $p\in\Fvis{G-e}{u_1}$ nor $p\in\Fvis{G-e}{u_2}$.

% ==========================================================
\section{$\NP$-completeness of $\BTD$}\label{sec:NPC}

Technically, in order to prove $\NP$-completeness of $\BTD$ we perform a polynomial-time reduction from the classical version of the Minimum hitting set ($\MHS$) problem.

\medskip
\indent\textsc{Minimum Hitting Set ($\MHS$)}\\
\indent\quad\textit{Input:} A set $\A = \{a_1, \ldots, a_n\}$, a family $\A_1, \ldots, \A_m$ of subsets of $\A$, an integer $t\ge0$.\\
\indent\quad\textit{Question:} Does there exist $\A' \subseteq \A$ such that $|\A'| \leq t$ and $\A' \cap \A_j \neq \emptyset$, $j \in \{1, \ldots, m\}$?

\medskip\noindent
In what follows we use the following example instance of $\MHS$: 
let $n=6$, $m=4$, $t=2$ and let $\A=\{a_1,\ldots,a_6\}$, $\A_1=\{a_1,a_2,a_6\}$, $\A_2=\{a_2,a_4,a_5,a_6\}$, $\A_3=\{a_1,a_3,a_5\}$, $\A_4=\{a_3,a_4\}$.

%...................................................................
\subsection{Construction and terminology}\label{sec:NPC_constr}

In this section, on the basis of the input to the $\MHS$ problem, we construct an appropriate forest $F$ consisting of the tree $T$, called the \emph{main} component, and some number of \emph{additional} connected components created on the basis of 'template' trees $\Tbard$ that we define later on.

In what follows we extensively use specific star subgraphs. By $S_n$ we denote an $n$-vertex star, while for a vertex $v$ in $G$ we use $S(v)$ to denote the subgraph of $G$ induced by $v$ and its neighbors $u$ for which $\dg{G}{u}=1$. 
An important parameter in our construction is an integer $\M$ that we calculate at the end of this section such that it is 'large enough'---the role of $\M$ is explained in the next section.

First, we focus on the structure of the main component $T$. 
The tree $T$ can be obtained by the identification of distinguished vertices $r_i$ of the trees $T(a_i)$ carefully constructed for the corresponding elements $a_i\in \A$, $i\in\{1,\ldots,n\}$.
The resulting common vertex, denoted by $r$, becomes the root of $T$.
Next, we add $\M+3(m-1)+4$ edges incident to $r$ and hence $S(r)$ is a star $S_{\M+3(m-1)+5}$.
As the building blocks of $T(a_i)$ we need graphs $\Gk$ with $\ka\ge\M+1$ and $\Gbi$.
The structure of these graphs, with distinguished vertices called \emph{connectors}, is presented in Figure~\ref{fig:gadgets}.
The graph $\Gk$ has one connector denoted by $w_1$ and $\Gbi$ has two connectors $r_i$ and $v_i$. The `loop' at a vertex depicts the star  with the distinguished central vertex.

   \begin{figure}[htb!]
    \begin{center}
    \includegraphics[scale=0.7]{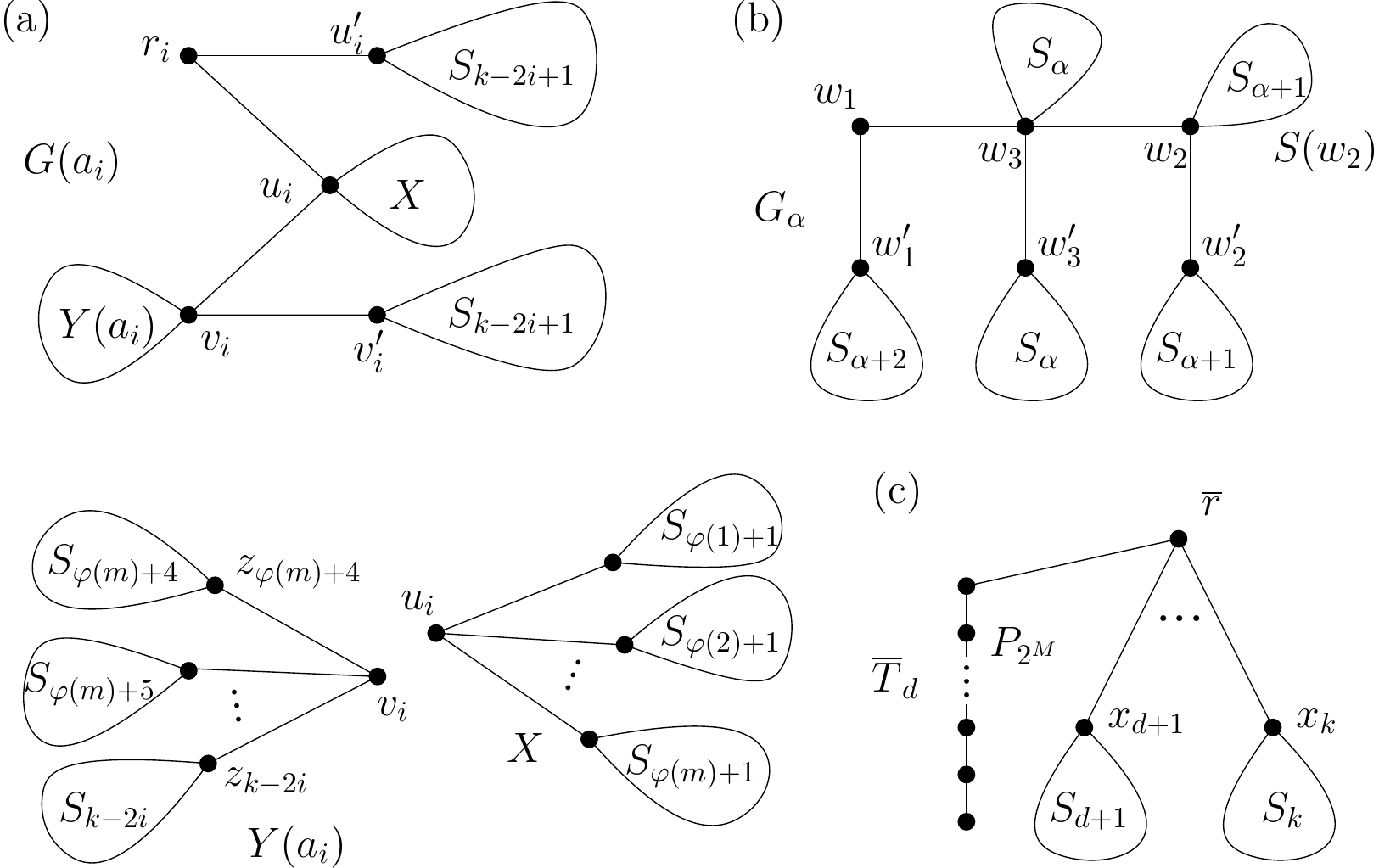}
    \caption{The graphs used in Construction: (a) $\Gbi$, (b) $\Gk$, (c) $\Tbar_d$, where $\M\le d<\ka$.}
    \label{fig:gadgets}
   \end{center}
    \end{figure}

From a slightly informal perspective, we can describe $T(a_i)$ as a 'star-shaped' structure formed with $m+1$ graphs $T_0,\ldots,T_m$, where $T_0=\Gbi$ and $T_j=G_{\varphi(j)}$ with
\[\varphi(j)=\M+3(j-1)+1\] for every $j\in\{1,\ldots,m\}$.
More formally, $T(a_i)$ is formed by the identifications of appropriate connectors, i.e., the connector $w_1$ of each $T_j$ with $j\in\{1,\ldots,m\}$ is identified with the connector $v_i$ of $T_{0}$.
In order to complete the construction of $T(a_i)$, for every $\A_j$ such that $a_i\in \A_j$ we take a copy of an $2^{\M}$-vertex path $\PM$ and identify an end-vertex of the path with a leaf of $S(w_2)$ of the corresponding tree $T_j$ in $T(a_i)$ (we say that $\PM$ is \emph{attached} at $T_j$).

At this point, fixing $\M$, we remark on the structure of $T$. 
Concerning the structure of $\Gbi$, we stress out that this depends on $a_i$ which determines the degrees of $u'_i, v'_i$ as well as the degree of $v_i$ and hence the number of vertices in $\{z_{\varphi(m)+4},\ldots,z_{k-2i}\}$ (see Figure \ref{fig:gadgets}(a)).
Also note that for a given $a_i\in \A$ the trees $T_1,\ldots,T_m$ in $T(a_i)$ are pairwise different and their structure depends on $j$. On the other hand the trees $T_1,\ldots,T_m$ are independent of $a_i$.
However, for distinct elements $a_{i_1},a_{i_2}\in \A$ the structure of $T(a_{i_1})$ and $T(a_{i_2})$ is different, which follows from the aforementioned differences between  $\Gb(a_{i_1})$ and $\Gb(a_{i_2})$ as well as varying 'attachment patterns' of paths $\PM$.

In order to complete the construction of the forest $\F$ we need the aforementioned trees $\Tbard$ with $d\ge\M$, defined as in Figure \ref{fig:gadgets}(c) (for the distinguished path $\PM$ we say that it is \emph{attached} at $\Tbard$). 
Namely, the forest $\F$ is composed of a single copy of $T$, $2n-|\A_j|+1$ copies of $\Tbar_{\varphi(j)-1}$ for $j=1$, and $n-|\A_j|+1$ copies of $\Tbar_{\varphi(j)-1}$ for each $j\in\{2,\ldots,m\}$, as well as $n-1$ copies of $\Tbar_{\varphi(j)}$ and $n$ copies of $\Tbar_{\varphi(j)+1}$ both templates taken for each $j\in\{1,\ldots,m\}$.
Consequently, with reference to our example of the $\MHS$ problem (recall $n=6$, $m=4$) we get, in order,
a unique copy of $T$, 10 copies of $\Tbar_\M$, 
3 copies of $\Tbar_{\M+3}$, 4 copies of $\Tbar_{\M+6}$
and 5 copies of $\Tbar_{\M+9}$,
5 copies of $\Tbar_{\M+4}$, $\Tbar_{\M+7}$ and $\Tbar_{\M+10}$,
6 copies of $\Tbar_{\M+5}$, $\Tbar_{\M+8}$ and $\Tbar_{\M+11}$.
The role of particular components we explain in the next section.

In what follows, referring to a path $\PM$ we always mean one of the paths $\PM$ used in this construction. 
For each $j\in\{1,\ldots,m\}$ the set of all paths $\PM$ corresponding to the set $\A_j$ (i.e. attached at $T_j$ in some $T(a_i)$ or at an additional component $\Tbard$ with $d\in\{\varphi(j)-1,\varphi(j),\varphi(j)+1\}$) is denoted by $\mcP(\A_j)$. 
Let $\mathcal{P}^* = \mcP(\A_1)\cup\cdots\cup\mcP(\A_m)$ and let $\mbar$ denote the size of the set $W$ consisting of the edges of the forest $\F$ excluding the edges of paths in $\mcP^*$. More formally 
\begin{equation}\label{eq:mbar}
    \mbar = \vert W \vert = \vert E(\F) \setminus \bigcup{E(\PM)}\vert, 
              \mbox{ where the sum runs over all } \PM\mbox{ in }\mcP^*.
\end{equation}
We remark that $\mbar$ depends on $k$ because the construction of $\F$ depends on $k$. For the $\BTD$ problem, we set the input parameters $\B$ and $k$ to be:
\begin{align}
       k &= \M + 1 + 3m + 2n + t,  \label{eq:k}\\ 
      \B &= n2^\M +\mbar,          \label{eq:b}
\end{align}
where $\M$ is chosen as a minimum integer satisfying  
\begin{equation}\label{eq:M}
      \M\geq 2\lceil\log_2\mbar\rceil+1.
\end{equation}
Note that \eqref{eq:M} can be rewritten as a bound $2^{\M}>4\mbar^2$ (we will use this form in the proof).
We remark that the above parameters can be set in such a way that their values are polynomial in $m,n$ and $t$.
It follows from the construction that $\mbar$ is polynomially bounded in $k$, that is, $\mbar\leq c_1k^c$ for some constants $c_1$ and $c$.
Take $\M$ and $k$ so that $\M\geq 2c\log_2k+2\log_2c_1+3$.
This can be done in view of~\eqref{eq:k}.
Note that this fixes the values of $\mbar$ and $\B$.
The right hand side of~\eqref{eq:M} is upper bounded by $2c\log_2k+2\log_2c_1+3$, which implies that the required bound in~\eqref{eq:M} holds.
This bound on $\M$ in particular implies $k=O(m+n+t)$ and thus $\mbar$, $2^{\M}$ and $\B$ are polynomial in $m,n$ and $t$.

%.............................................................
\subsection{The idea of the proof}\label{sec:NPC_idea} 

First of all, we note that the vast majority of edges in the forest $F$ belongs to the additional components, which due to their structure fit into precisely planned levels of an elimination forest $\ET{F}$, thus leaving exactly the right amount of capacity on those levels where the main component gadgets come into play. 
The positioning of elimination subtrees corresponding to appropriate paths $\PM$ is determined in Lemma \ref{lem:lem11}; also see Figure \ref{fig:PMlevels} where we sketch a pattern of fitting  elimination subtrees into appropriate levels.
Though most of the capacity consumed by the main component can be  attributed to the paths $\PM$ (attached at leaves of respective instances of the gadget $\Gk$) the role of just a few edges of $\Gbi$ and $\Gk$ cannot be overestimated. 
 \begin{figure}[htb!]
    \centering
    \includegraphics[scale=0.7]{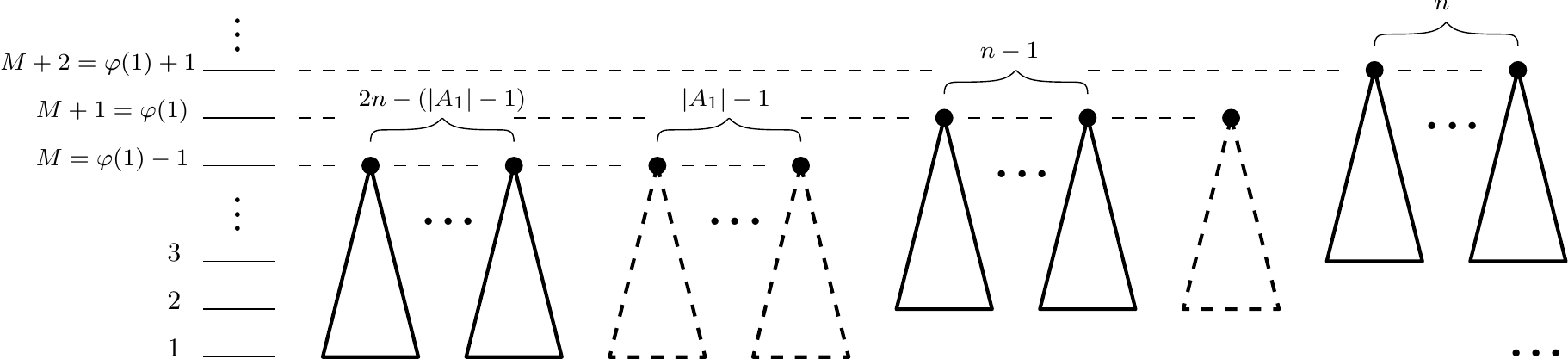}
    \caption{Positioning of elimination subtrees corresponding to the paths $\PM$ in the main and additional components (dashed and solid lines, resp.) of the forest $F$. A snapshot for $j=1$.}
    \label{fig:PMlevels}
\end{figure}
As we will see, the assignment of the edge $\{u_i,v_i\}$ to the level $k-2i+2$ is equivalent to including the element $a_i$ in the solution $\A'$. 
Due to the sensitivity of the gadget $\Gk$ to what levels are visible from its connector $w_1$ 'outside' $\Gk$ (see Lemmas \ref{lem:lem3a} and \ref{lem:lem3b}) we get a coupling between the level of $\{u_i,v_i\}$ and the highest level that can be occupied by the so called 'root edge' of the path $\PM$ attached at $\Gk$, when particular instance $G_{\varphi(j)}$ of the gadget in $T(a_i)$ corresponds to the set $\A_j$ containing $a_i$. 
More specifically, if $\lev{}{\{u_i,v_i\}}=k-2i+2$, then the root edge of such a path $\PM$ can be moved to $\varphi(j)$ from $\varphi(j)-1$ allowing all of its other edges to be moved one level up, thus gaining the increase of the free space at the level that has been previously occupied by roughly half of its edges.
 
The bound $t$ on the size of solution $\A'$ is met by attaching appropriate number of edges pending at the root $r$ of the main component $T$. The number of such edges depends on $t$ and it is calculated in such a way that in at most $t$ of $n$ subgraphs $T(a_i)$ the edge $\{r,u_i\}$ will be allowed at level not exceeding $k-2n$. 
In Lemma \ref{lem:lem6} we show that either $\{r,u_i\}$ or $\{u_i,v_i\}$ must occupy the level $k-2i+2$ and hence there will be at most $t$ subgraphs $T(a_i)$ with the edge $\{u_i,v_i\}$ assigned to the level $k-2i+2$ and triggering the above-mentioned process of lifting.

%............................................................
\subsection{Proof of $\NP$-completeness}\label{sec:NPC_prf} 

We start with a simple lemma on the low levels of elimination trees that corresponds to one of the known basic properties of edge rankings.

\begin{lemma}\label{lem:basic}
Let $G$ be a connected graph. For every elimination tree $\ET{G}$ there exists an elimination tree $\ET{G}'$ with $h(\ET{G}')=h(\ET{G})$ such that for each $S(v)$ in $G$ the edges of $S(v)$ belong to the levels in $\{1,\ldots,|E(S(v))|\}$.
\end{lemma}

Now, we present several lemmas on properties of the gadgets defined in Section \ref{sec:NPC_constr}.
We always treat the gadgets as if they were subgraphs of the forest $F$, e.g., in the proofs of the next two lemmas, when analyzing admissibility of levels for particular edges of $\Gk$, we use $\Fvis{F'}{w_1}$ to refer to the levels visible from $w_1$ in $F'$, where $F'$ is a graph induced by $V(F)\setminus V(\Gk-w_1)$ (see Figure \ref{fig:Gadg_color}).
Moreover, in what follows, both lemmas are used in the context in which the levels $\ka+4,\ldots,k$ belong to $\Fvis{F'}{w_1}$.
However, since that context is not required for their proofs, we just mention this fact and proceed with more general statements.
Depending on $\Fvis{F'}{w_1}$, the following two lemmas characterize the cases encountered upon construction of elimination trees $\ET{\Gk}$.

   \begin{figure}[htb!]
    \begin{center}
    \includegraphics[scale=0.7]{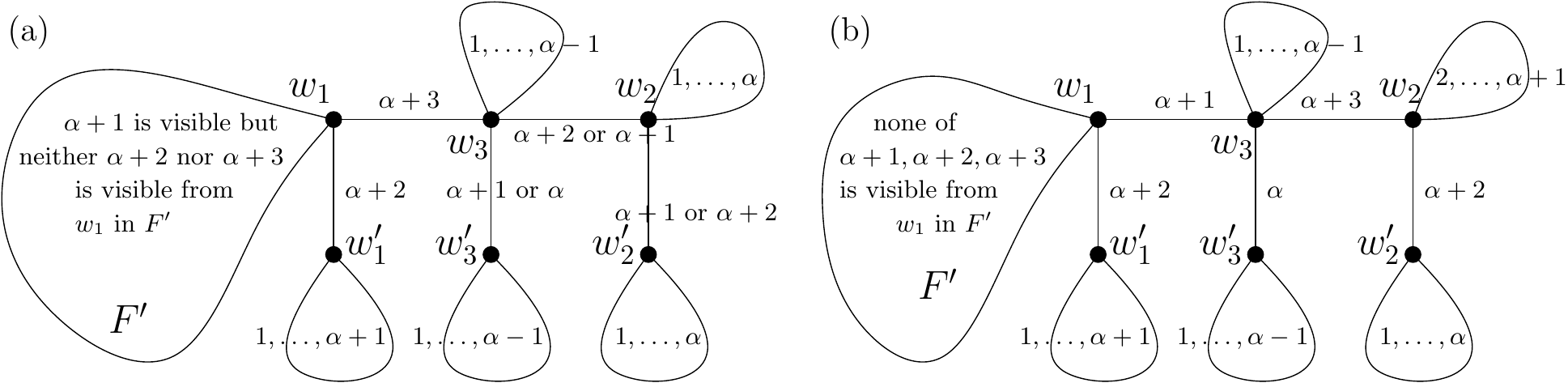}
     \caption{An illustration of the proofs: 
      (a) for Lemma \ref{lem:lem3a}, (b) for Lemma \ref{lem:lem3b}.} 
    \label{fig:Gadg_color}
   \end{center}
    \end{figure}

\begin{lemma}\label{lem:lem3a}
 If $\ka+2$ or $\ka+3$ belongs to $\Fvis{F'}{w_1}$, then for every elimination tree $\ET{\Gk}$ it holds $h(\ET{\Gk})>\ka+3$. If neither $\ka+2$ nor $\ka+3$ belongs to $\Fvis{F'}{w_1}$ and $\ka+1\in\Fvis{F'}{w_1}$, then in \emph{every} elimination tree $\ET{\Gk}$ of height $\ka+3$ the levels admissible for $S(w_2)$ are in  $\{1,\ldots,\ka\}$, and the only levels visible from $w_1$ in $\Gk$ are $\ka+2$ and $\ka+3$.
\end{lemma}

\begin{proof}
Let $H'$ and $H$ be the subgraphs induced by $V(F')\cup\{v\,|\,\{w'_1,v\}\in E(\Gk)\}$ and $V(H')\cup\{w_3\}$, respectively.

Suppose that, contrary to our claim, there exists an elimination tree $\ET{\Gk}$ with $h(\ET{\Gk})\le\ka+3$. Let $D=\{1,\ldots,\ka+3\}$ be a set of levels admissible for $\ET{\Gk}$ and let $p\in\{\ka+2,\ka+3\}$ be a level visible from $w_1$ in $F'$. Clearly, the levels of adjacent edges $\{w_1,w'_1\}$ and $\{w_1,w_3\}$ cannot be the same, and due to visibility restrictions have to be different from $p$. It is not hard to see that by Lemma~\ref{lem:basic} independently of whether $p$ equals $\ka+2$ or $\ka+3$, the levels visible from $w_1$ in $H'$ form the following set $\{a,\ldots,\ka+3\}$, where $a=\lev{}{\{w_1,w'_1\}}$. Thus, the set of levels admissible for $\{w_1,w_3\}$ is $D\setminus\{a,\ldots,\ka+3\}$ and hence $\lev{}{\{w_1,w_3\}} < a$. Therefore, $\lev{}{\{w_1,w_1'\}}$ as well as  $p$ and $\lev{}{\{w_1,w_3\}}$ are visible in $H$ from $w_3$, which means that there are at most $|D|-3$ levels admissible for the edges in $U=\{\{w_3,v\}\,|\,v\notin V(H)\}$. Since edges in $U$ are pairwise adjacent and $|U|=\ka+1$, at least one edge in $U$ has to be assigned to a level higher than $\ka+3$, which contradicts our assumption on the height of $\ET{\Gk}$.

The proof of the second claim, consists of two parts. 
First, we prove that the highest level $\ka+3$ must be occupied by an edge of the subgraph $G''$ induced by $\{w_1\}\cup V(S(w_3))\cup V(S(w'_3))\cup V(S(w'_1))$. Since the highest level is always a singleton, we will deduce that the height of an elimination tree $\ET{G'}$ is at most $\ka+2$, where $G'$ is a subgraph induced by $\{w_3\}\cup V(S(w_2))\cup V(S(w'_2))$. Next, restricting our attention to $G'$ we show that in every $\ET{G'}$ the levels admissible for the edges of $S(w_2)$ are in $\{1,\ldots,\ka\}$. 

Suppose that there exists $\ET{G''}$ such that $h(\ET{G''})\le\ka+2$. We start with an observation that $\lev{}{\{w_1,w'_1\}}=\ka+2$, for otherwise if $\lev{}{\{w_1,w'_1\}}\le\ka$ (recall visibility of level $\ka+1$ from $w_1$ in $F'$), then at least one edge of $S(w'_1)$ would have to occupy level $\ka+3$, which contradicts our assumption. 
Similarly, $\lev{}{\{w_1,w_3\}}>\ka$, for if not, at least $3$ levels would be visible from $w_3$ in $H$ and hence $\ka$ pairwise adjacent edges in the set $U=\{\{w_3,v\}\,|\,v\notin V(H)\}$ would have to be assigned to at most $\ka-1$ admissible levels. Hence $\lev{}{\{w_1,w_3\}}>\ka$, and since the levels in $\{\ka+1,\ka+2\}$ are visible from $w_3$ in $H'$, we get $\lev{}{\{w_1,w_3\}}=\ka+3$, a contradiction. Therefore since $|\Lev{\ka+3}{\ET{\Gk}}|=1$, it holds $h(\ET{G'})\le\ka+2$.
Note that we also proved the statement on visibility of $\ka+2$ and $\ka+3$ from $w_1$ in $\Gk$.

Now, for the second part we focus on $G'$ and observe that $\ka+2$ edges incident to $w_2$ have to occupy distinct levels in $\{1,\ldots,\ka+2\}$. Thus, the levels visible from $w_3$ in $G'$ form the following set $\{a,\ldots,\ka+2\}$, where $a=\lev{}{\{w_2,w_3\}}$. Consequently, there are at most $a$ levels admissible for $\ka+1$ pairwise adjacent edges in $U'=\{\{w_3,v\}\,|\,v\notin V(G')\}$, which implies $a\ge\ka+1$. 
It remains to observe that if $a=\ka+1$, then the level $\ka+2$ cannot be occupied by an edge of $S(w_2)$, for if not, both $\ka+1$ and $\ka+2$ 
would be visible from $w'_2$ and hence forbidden in $S(w'_2)$ leaving just $\ka$ levels for $\ka+1$ edges incident to $w'_2$. 
The argument for $a=\ka+2$ is analogous. 
\end{proof}

\begin{lemma}\label{lem:lem3b}
 If neither $\ka+1$, $\ka+2$ nor $\ka+3$ belongs to $\Fvis{F'}{w_1}$, then there \emph{exists} an elimination tree $\ET{\Gk}$ of height $\ka+3$ with all edges of $S(w_2)$ at levels in $\{2,\ldots,\ka+1\}$ and such that only the levels $\ka+1,\ka+2$ and $\ka+3$ are visible from $w_1$ in $\Gk$. Moreover, there does not exist an elimination tree $\ET{\Gk}$ with $h(\ET{\Gk})<\ka+3$.
\end{lemma}

\begin{proof}
Concerning the first claim, an assignment of the edges in $E(\Gk)$ to appropriate levels of $\ET{\Gk}$ is presented in Figure~\ref{fig:Gadg_color}(b).
For the second claim suppose that, on the contrary, there exists an elimination tree $\ET{\Gk}$ such that $h(\ET{\Gk})\le\ka+2$.
The argument is analogous to that in the proof of Lemma \ref{lem:lem3a}.
For $H'$ and $H$ denoting the graphs induced by $V(F')\cup\{v\,|\,\{w'_1,v\}\in E(\Gk)\}$ and $V(H')\cup\{w_3\}$, respectively, we easily observe that if $a$ stands for $\lev{}{\{w_1,w'_1\}}$, then a set of levels visible from $w_1$ in $H'$ is $\{a,\ldots,\ka+2\}$ and hence $\lev{}{\{w_1,w_3\}}<a$.
Clearly, due to the visibility from $w_3$ in $H$ there are at most $\ka$ levels admissible for the edges in $U=\{\{w_3,v\}\,|\,v\notin V(H)\}$.
Therefore since $|U|=\ka+1$, at least one edge in $U$ has to be assigned to a level higher than $\ka+2$, a contradiction.
\end{proof}

The two lemmas that follow are devoted to the description of the mutual interaction of the distinguished edges of $\Gbi$, which manifests as a 'switching property' of the gadget and allows 'lifting' of appropriate elimination subtrees preserving the bound $\B$ on the level size. 

\begin{lemma}\label{lem:lem6}
If $\ET{T}$ is an elimination tree of height $k$, then for each $i\in\{1,\ldots,n\}$ 
\begin{equation*}
    either \quad \lev{}{\{r,u_i\}}=k-2i+2 \quad 
    or     \quad \lev{}{\{u_i,v_i\}}=k-2i+2,
\end{equation*}
and each level $p\in\{k-2i+1,\ldots,k\}$ is visible from $r$ in the subgraph induced by the vertices of $T(a_1),\ldots,T(a_i)$,
\end{lemma}
\begin{proof}
 For $i=1$ we consider an elimination tree $\ET{H}$, where $H=T(a_1)$. 
 Considering the visibility of levels from $v_1$ in the subgraph induced by $V(S(v'_1))\cup\{v_1\}$, by Lemma \ref{lem:basic} we conclude that there is no loss of generality in assuming that the levels assigned to the edges of $S(v'_1)$ belong to $\{1,\ldots,k-2\}$ and $\lev{}{\{v_1,v'_1\}} = k-1$. 
 Similarly, $\lev{}{\{r,u'_1\}} = k-1$. Thus, %by \ref{prop:p2}
 either $\lev{}{\{r,u_1\}}>k-1$ or $\lev{}{\{u_1,v_1\}}>k-1$. Consequently, $h(\ET{H})=k$ and both $k-1$ and $k$ are visible from $r$ in $H$.
 Now, suppose our lemma holds for $\ET{H}$, where $H$ is a subgraph of $T$ induced by the vertices of $T(a_1),\ldots,T(a_{i-1})$. In order to extend $\ET{H}$ to $\ET{H'}$, where $H'$ is a subgraph induced by the vertices of $T(a_1),\ldots,T(a_i)$ we use an analogous argument as for $i=1$. 
 Thus, $\lev{}{\{v_i,v'_i\}} = \lev{}{\{r,u'_i\}} = k-2i+1$ and hence at least one of the edges $\{u_i,v_i\}$, $\{r,u_i\}$ must be assigned to a level $p > k-2i+1$. Therefore since by assumption all levels $p\in\{k-2i+3,\ldots,k\}$ are visible in $H$ from $r$, either $\{u_i,v_i\}$ or $\{r,u_i\}$ must be assigned to the level $k-2i+2$. 
 Thus $h(\ET{H'})=k$, and all levels $p\in\{k-2i+1,\ldots,k\}$ are visible in $H'$ from $r$.
\end{proof}

\begin{corollary}\label{cor:lem6}
 If $\ET{T}$ is an elimination tree of height $k$, then each level in $\{k-2n+1,\ldots,k\}$ is visible in $T$ from $r$ and if $\lev{}{\{u_i,v_i\}}=k-2i+2$, then $\lev{}{\{r,u_i\}}\le k-2n$, $i\in\{1,
 \ldots,n\}$.
\end{corollary}

In fact, the above lemma tells a little bit more. 
Namely, if $T[v]$ denotes the subtree of $T$ rooted at $v$ and induced by $v$ and all vertices having $v$ as an ancestor, then whenever the edge $\{r,u_i\}$ (the edge $\{u_i,v_i\}$) gets assigned to the level $k-2i+2$, then the levels $k-2i+2,\ldots,k$ become visible from $u_i$ even in the subgraph induced by $V(T)\setminus(V(T[u_i])\setminus\{u_i\})$ (from $v_i$ in the subgraph induced by $V(T)\setminus(V(T[v_i])\setminus\{v_i\})$) and hence become forbidden for the edges of $T[u_i]$ (the edges of $T[v_i]$).
Accordingly, in our next lemma we assume that $\{1,\ldots,k-2i+1\}$ is a set of levels admissible for $T[u_i]$ and $T[v_i]$.
It is also worth mentioning that if $\lev{}{\{r,u_i\}}=k-2i+2$, then due to the visibility of particular levels, the edge $\{u_i,v_i\}$ must be assigned to a level $p$ satisfying $p<\M$, which is crucial in our reduction. In fact, pushing $\{u_i,v_i\}$ to a low level triggers the above-mentioned 'switching' by making certain levels in the subgraph $X$ visible from the vertex $v_i$.   
In particular, for each $j\in\{1,\ldots,m\}$ the level $\varphi(j)+1$ becomes visible in $X$ from the vertex $w_1\in V(G_{\varphi(j)})$ which by Lemma~\ref{lem:lem3a} results in impossibility of 'lifting' an elimination tree $\ET{\PM}$, where $\PM$ is attached at $G_{\varphi(j)}$.

\begin{figure}[htb!]
    \centering
    \includegraphics[scale=0.72]{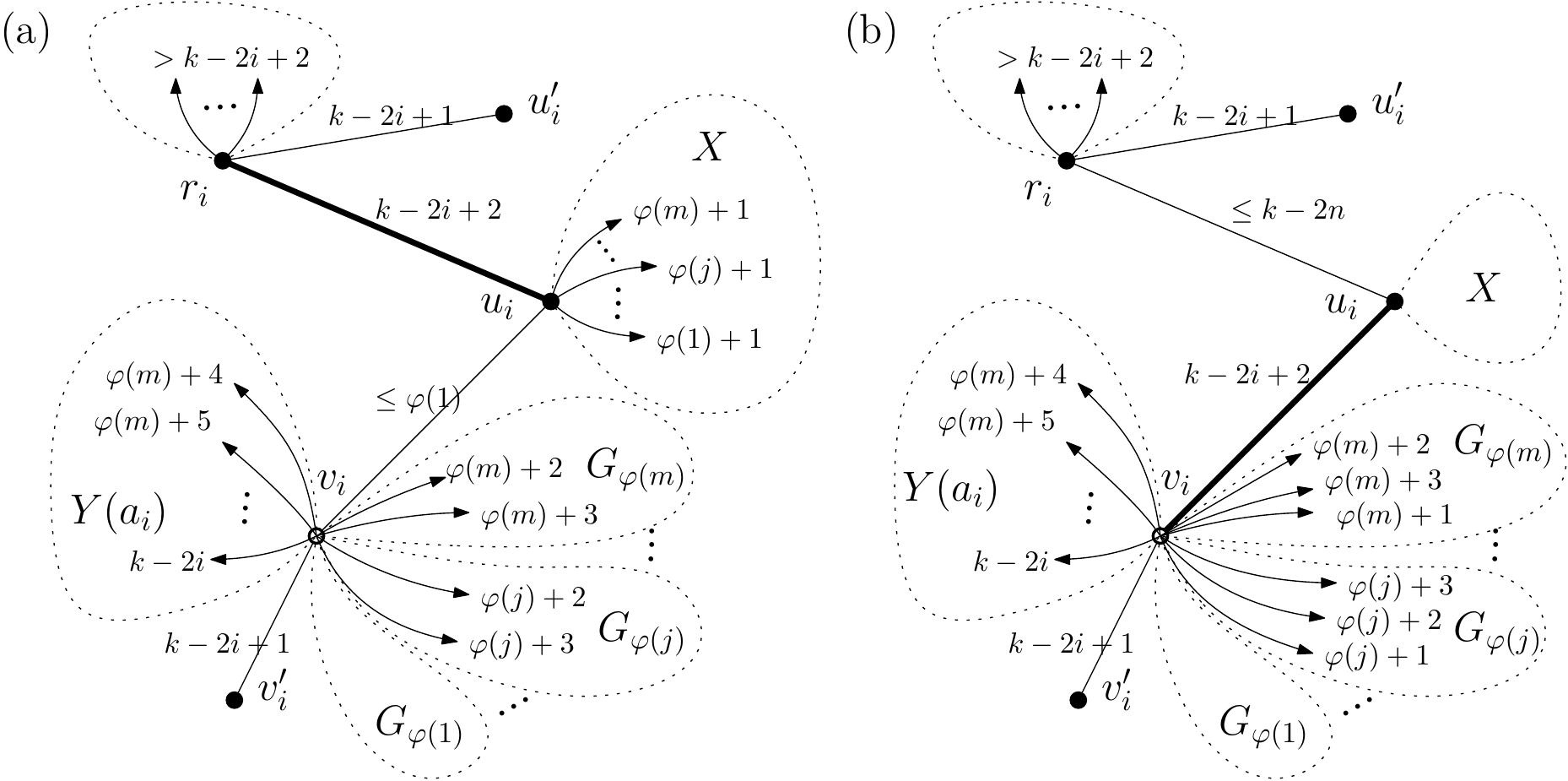}
    \caption{An illustration for the proof of Lemma~\ref{lem:lem9}. The two major cases of visibilities. The arrows point to the levels visible from $v_i$.}
    \label{fig:Gbis_color}
\end{figure}

\begin{lemma}\label{lem:lem9}
  Let $\ET{T}$ be an elimination tree of height $k$. For every $i\in\{1,\ldots,n\}$:
\begin{enumerate} [label={\normalfont(\alph*)}]
      \item\label{it:9a} if $\lev{}{\{r,u_i\}}=k-2i+2$, then \emph{in every} $\ET{T[u_i]}$ of height $k-2i+1$ for each $T_j$ with $j\in\{1,\ldots,m\}$ it holds $\lev{}{e}\in\{1,\ldots,\varphi(j)\}$, where $e$ is an edge of $S(w_2)$ in $T_j$.
      \item\label{it:9b} if $\lev{}{\{u_i,v_i\}}=k-2i+2$, then there \emph{exists} an elimination tree $\ET{T[v_i]}$ of height $k-2i+1$ such that in each $T_j$ with $j\in\{1,\ldots,m\}$  it holds $\lev{}{e}\in\{2,\ldots,\varphi(j)+1\}$, where $e$ is an edge of $S(w_2)$ in $T_j$. Moreover  $\lev{}{e}\le\varphi(j)+1$ in every $\ET{T[v_i]}$.
  \end{enumerate}
\end{lemma}

\begin{proof}
First we prove \ref{it:9a}. 
As we have already mentioned, the levels higher than $k-2i+2$ are not admissible for $T[u_i]$.
%Following property \ref{prop:p1}, 
Hence, if $\lev{}{\{r,u_i\}}=k-2i+2$, then $\lev{}{\{u_i,v_i\}}\notin \{k-2i+1,k-2i+2\}$, since following the argument in the proof of Lemma \ref{lem:lem6} both levels are already occupied by $\{v_i,v'_i\}$ and $\{r,u_i\}$, the edges adjacent to $\{u_i,v_i\}$. 
Similarly, considering the visibility of levels from $v_i$ in the subgraph induced by $V(S(z_{\varphi(m)+4}))\cup\cdots\cup V(S(z_{k-2i}))\cup\{v_i\}$, from Lemma \ref{lem:basic} it follows that we can assume $\lev{}{\{v_i,z_p\}}=p$, where $p\in\{\varphi(m)+4,\ldots,k-2i\}$ and that none of these levels is admissible for $\{u_i,v_i\}$.
Summarizing, we get 
\begin{equation}\label{eq:phi3}
      \lev{}{\{u_i,v_i\}}\leq \varphi(m)+3.
\end{equation}

Now, in order to prove $\lev{}{\{u_i,v_i\}}\le\varphi(1)$, we use induction on $j$ going downwards from $m$ to $1$. The argument for the base case, i.e., $j=m$ and for the inductive step is the same. So let us take an arbitrary $j\in\{m-1,\ldots,1\}$ and assume that the inductive hypothesis holds for $m,\ldots,j+1$. 
More precisely, for each parameter $l\in\{j+1,\ldots,m\}$ it holds $\lev{}{\{u_i,v_i\}}\le\varphi(l)$ and all levels in $\{\varphi(l)+1,\varphi(l)+2,\varphi(l)+3\}$ are visible from $v_i$ in the subgraph induced by $V(T)\setminus(V(G_{\varphi(1)})\cup\cdots\cup V(G_{\varphi(l)}))$ (in fact, an argument similar that we used for \eqref{eq:phi3} shows that each level $p\ge\varphi(l)+1$ is visible).

For an inductive step suppose that, on the contrary, $\lev{}{\{u_i,v_i\}}>\varphi(j)$. Naturally,  $\lev{}{\{u_i,v_i\}}\neq\varphi(j)+1$, since there is an edge at level $\varphi(j)+1$ visible from $u_i$ in $X$. Hence, it remains to consider the assignment of $\{u_i,v_i\}$ either to $\varphi(j)+2$ or $\varphi(j)+3$. 
However, in both cases by Lemma \ref{lem:lem3a} we get $h(\ET{G_{\varphi(j)}})>\varphi(j)+3$ which means that there is at least one edge of $G_{\varphi(j)}$ at level higher than $\varphi(j)+3$ and visible in $G_{\varphi(j)}$ from $v_i$. 
A contradiction, since by the inductive assumption no level in $\{\varphi(j+1),\ldots,k\}$ is admissible for the edges of $G_{\varphi(j)}$. 
Thus $\lev{}{\{u_i,v_i\}}\le\varphi(j)$. 
An important consequence of $\lev{}{\{u_i,v_i\}}\le\varphi(j)$ is that $\varphi(j)+1$ is visible from $v_i$ in $X$. Moreover, $h(G_{\varphi(j)})=\varphi(j)+3$ (by assumption and the lower bound in Lemma \ref{lem:lem3b}). 
Thus we are allowed to use Lemma \ref{lem:lem3a} from which it follows that the levels admissible for the edges of $S(w_2)$ in $G_{\varphi(j)}$ belong to $\{1,\ldots,\varphi(j)\}$ and that both $\varphi(j)+2$ and $\varphi(j)+3$ are visible in $G_{\varphi(j)}$ from $v_i$. 
Hence, each level $p\in\{\varphi(j)+1,\ldots,k\}$ is visible from $v_i$ in the subgraph induced by $V(T)\setminus(V(G_{\varphi(1)})\cup\cdots\cup V(G_{\varphi(j)}))$.

Next, we prove \ref{it:9b}. 
Similarly as for part \ref{it:9a} the analysis of visibilities from the vertex $v_i$ conducted for particular subgraphs of $T$ and Lemma \ref{lem:lem6} let us conclude that $\{\varphi(m)+4,\ldots,k\}$ is a set of levels visible from $v_i$ in the subgraph induced by $V(T)\setminus(V(G_{\varphi(1)})\cup\cdots\cup V(G_{\varphi(m)}))$.
The proof is completed by Lemma \ref{lem:lem3a} applied to each $G_{\varphi(j)}$ with $j\in\{1,\ldots,m\}$. 
\end{proof}

\begin{lemma}\label{lem:lem10}
 For every tree $\Tbard$ with $d\ge\M$, there exists an elimination tree $\ET{\Tbard}$ of height $k$ such that for each edge $e$ of $\PM$ attached at $\Tbard$ it holds $\lev{}{e}\le{d}$. Moreover, if $\ET{\Tbard}$ is of minimum height $k$, then $\lev{}{e}\ge d-\M+1$.
\end{lemma}

\begin{proof}
   As we already know from Lemma \ref{lem:basic} there is no loss of generality in assuming that for every $p\in\{d+1,\ldots,k\}$ we have $\lev{}{\{\rbar,x_p\}}=p$ and hence the levels in $\{d+1,\ldots,k\}$ are visible from $\rbar$ in $\Tbar_d$. Consequently, no edge of $\PM$ can be assigned to a level higher than $d$. It remains to note that the minimum height elimination tree of the line graph of $2^{\M}$-vertex path is the $\M$-level full binary tree (see, e.g., \cite{KatchMcCauSeg95}). 
\end{proof}

% --------------------------------------------------------- bounded case
%
Before the proof of Theorem \ref{thm:NPC_forests} we need to consider several properties of the forest $F$ when elimination tress are assumed to be $\B$-bounded.
Let $r(\PM)$ denote the \emph{root edge} of $\PM$, i.e. the edge assigned to the highest level of $\ET{\PM}$ and let 
$$
R=\{r(\PM)\,|\,\PM\in\mcP(\A_j) \mbox{ for } j\in\{1,\ldots,m\}\}.
$$
We will also need a partition $(R_m,R_a)$ of $R$ with $R_m$ and $R_a$ being the sets of the root edges of the paths attached at the main and additional components, respectively.

Concerning the role of Lemmas~\ref{lem:lem9} and \ref{lem:lem10} we note that they provide the highest possible levels at which the root edges of paths $\PM$ in $\F$ can be placed in $\ET{F}$.
However, they may be placed potentially at much lower levels.
We rule, to some extent, this possibility in the next lemma.

\begin{lemma}\label{lem:lem11}
  Let $F$ be the forest corresponding to an instance of the $\MHS$ problem and let $\ET{F}$ be $\B$-bounded.
  For each $i\in\{1,\ldots,n\}$ and $j\in\{1,\ldots,m\}$ the following properties hold.
   \begin{enumerate} [label={\normalfont(\alph*)},leftmargin=*,itemsep=0pt]
    \item\label{it:11a}
       If $\PM$ is attached at $T_j$ of $T(a_i)$, then
       $\lev{}{r(\PM)}\in\{\varphi(j)-1,\varphi(j)\}$.
    \item\label{it:11b}
       If $\PM$ is attached at $\Tbard$ with $d=\varphi(j)-1$,
       then $r(\PM)\in\Lev{d}{\ET{\F}}$.
    \item\label{it:11c}
      $\vert \Lev{\varphi(j)-1}{\ET{\F}}\cap R \vert\le \eta$,
      where $\eta=2n$ if $j=1$, and $\eta=n$
      for each $j\in\{2,\ldots,m\}$.
\end{enumerate}
\end{lemma}
\begin{proof} 
We prove by a contradiction that \ref{it:11c} holds.
We first consider $j=1$.
So suppose that $\vert \Lev{\M}{\ET{\F}}\cap R \vert \geq 2n+1$.
For every $\PM$ whose root edge is assigned to level $\M$, it holds that $\ET{\PM}$ has height $\M$ and thus it is a full binary tree having in particular $2^{\M-1}$ vertices at level $1$ in $\ET{F}$.
Thus we count the total number of such vertices at level $1$ resulting from $2n+1$ paths $\PM$. By~\eqref{eq:M} and~\eqref{eq:b}, we get $(2n+1)2^{\M-1}=n2^{\M}+2^{\M-1}=\B-\mbar+2^{\M-1} >\B$, a contradiction (there is more than $\B$ vertices at level $1$ in $\ET{F}$).

Consider a level $s\geq 1$. 
We will argue that at levels in $Q=\{1,\ldots,s+\M-1\}$ there may be at most $n+ns$ vertices from $R$ in total.
For proving this we use the following observation.  Take an arbitrary $\PM$.
Since any $\ET{\PM}$ is a binary tree, the number of vertices of $\ET{\PM}$ within levels $1,\ldots,p$, for any $p\in\{1,\ldots,h(\ET{\PM})\}$, is minimized when $\ET{\PM}$ is a full binary tree.
So suppose that each $\ET{\PM}$ in $\ET{F}$ is a full binary tree (we will drop this assumption in a moment) and, on the contrary, assume that the levels in $Q$ contain at least $n+ns+1$ vertices from $R$.
If every level in $Q\setminus\{1,\ldots,\M\}$ contains $n$ roots from $R$ and the level $\M$ contains $2n$ roots from $R$, then such a placement of the roots is called \emph{even}.
Considering an even placement and the one extra $\PM$ (recall our assumption for a contradiction) having its root in $Q$, we obtain that the overall number of vertices in $\ET{F}$ in the levels $1,\ldots,s$ is at least $sn2^{\M}+2^{\M-1}$, where $2^{\M-1}$ comes from the extra path $\PM$.

The placement of the vertices in $R$ does not have to be even. However, we argue that it cannot 'drift' too much from an even placement.
In the even placement: if $\PM$ is attached to $\Tbard$, then without loss of generality $\lev{}{r(\PM)}=d$.
However, by Lemma~\ref{lem:lem10}, we have in this case $\lev{}{r(\PM)}\leq d$ in an arbitrary placement of the roots.
Hence, when going from even to an arbitrary placement, the level of $r(\PM)$ may only decrease.
In the even placement: if $\PM$ is attached at $T_j$ of $T(a_i)$, then $\lev{}{r(\PM)}\in\{\varphi(j)-1,\varphi(j)\}$.
By Lemma~\ref{lem:lem9}, for such a $\PM$ the upper bound holds for an arbitrary placement but in the worst case we may have $\lev{}{r(\PM)}=\varphi(j)$ for all such paths.
We do consider the situation when all such $\PM$'s have $\lev{}{r(\PM)}=\varphi(j)$ (we call them \emph{$1$-lifted}) because we aim at providing a lower bound on the number of vertices in the levels $1,\ldots,s$ in $\ET{F}$.
Thus for each $j\in\{1,\ldots,m\}$, in worst case, at most $n$ paths $\PM$ with $\lev{}{r(\PM)}=\varphi(j)-1$ in the even placement have their roots at the level $\varphi(j)$ in an arbitrary assignment, i.e., there are at most $n$ $1$-lifted paths.
Taking into account that relaxing the property that each $\PM$ corresponding to a full binary tree in $\ET{F}$ increases the number of edges in the levels $1,\ldots,s$ we obtain a lower bound on their number in this range:
\[\sum_{s'=1}^{s}\vert\Lev{s'}{\ET{F}}\vert \geq sn2^{\M}-\sum_{p\geq 0}2^{\M-3-p}+2^{\M-1} \geq sn2^{\M}-2^{\M-2}+2^{\M-1}=sn2^{\M}+2^{\M-2},\]
where we additionally used a property that the highest level $s+\M-1$ does not contain any $1$-lifted paths $\PM$. 
(The reason we make this restriction is to ensure that there are no $n$ $1$-lifted paths at the level $s+\M+2$ (they would be moved from the level $s+\M+1$ with respect to an even assignment) in order to account for the sum $\sum_{p\geq 0}2^{\M-3-p}$.)
Therefore, the average number of vertices within levels $1,\ldots,s$ is at least
\[n2^{\M}+\frac{1}{s} 2^{\M-2} =\B-\mbar+\frac{1}{s}2^{\M-2} > \B-\mbar+\frac{1}{\mbar}2^{\M-2} \geq \B,\]
where we have used $s\leq\mbar$.
Also, the latter inequality is due to $2^{\M}\geq 4\mbar^2$ that follows from \eqref{eq:M}.
Thus, we have proved that the average number of edges in a level of the selected range exceeds $\B$ and by the property of the average, some level does contain at least this average number of vertices---a contradiction.

Recall that the above reasoning is restricted to the levels $s+\M-1$ that do not have $1$-lifted paths $\PM$, i.e., it applies to each level $\varphi(j)+1$ and $\varphi(j)-1$.
Taking $s+\M-1=\varphi(j)-1$ together with (due to Lemmas~\ref{lem:lem9} and~\ref{lem:lem10}) the fact that the overall number of roots at the levels $\varphi(j')-1,\varphi(j'),\varphi(j')+1$, for $j'<j$, is $3(j-1)n+n$, we conclude that $\vert \Lev{\varphi(j)-1}{\ET{\F}}\cap R \vert\le n$, which completes the proof of \ref{it:11c}.

By taking $s+\M-1=\varphi(j)+1$ and combining it with Lemma~\ref{lem:lem10}, we have that none $\ET{\PM}$ for the paths $\PM$ attached to $\Tbard$ with $d=\varphi(j)-1$ can be rooted at a lower level, i.e, it cannot be $\lev{}{r(\PM)}<\varphi(j)-1$.
This proves \ref{it:11b} for $d=\varphi(j)-1$.

We finally prove \ref{it:11a}.
The right hand side inequality in \ref{it:11a} follows from Lemma~\ref{lem:lem9}.
In fact, the lemma implies a stronger bound: Lemma~\ref{lem:lem9}\ref{it:9a} upper-bounds $\lev{}{r(\PM)}$ by $\varphi(j)-1$ and Lemma~\ref{lem:lem9}\ref{it:9b} upper-bound $\lev{}{r(\PM)}$ by $\varphi(j)$ in any elimination forest $\ET{F}$. 
The left hand side inequality in \ref{it:11a} follows from the same reasoning as in the preceding paragraph for $s+\M-1=\varphi(j)+1$ (dropping $r(\PM)$ of a $\PM$ attached at $T_{j+1}$ would give the necessary number of roots within levels $1,\ldots,\varphi(j)+1$ to have our contradiction).
\end{proof}

\begin{theorem} \label{thm:NPC_forests}
The Bounded Tree-Depth problem is $\NP$-complete for line graphs of forests.
\end{theorem}
\begin{proof}

  Let $\A$, $t$ and $\A_1,\ldots,\A_m$ form an instance of the $\MHS$ problem, and let $\F$ and $\B$ be the forest and positive integer obtained for that instance according to the Construction described in Section \ref{sec:NPC_constr} (recall that $k=\M+3m+t+2n+1$).
  
  ($\Rightarrow$) We are going to argue that if there exists a $\B$-bounded  elimination forest $\ET{\F}$ such that $h(\ET{\F})\le k$, then there exists a solution $\A'$ to the $\MHS$ problem such that $|\A'|\leq t$.
 A solution to the $\MHS$ problem is defined as follows:
 \begin{equation}\label{eq:solution}
    a_i\in \A' \mbox{ if and only if } \lev{}{\{u_i,v_i\}}=k-2i+2,
 \end{equation}
 for each $i\in\{1,\ldots,n\}$.
  First we prove that $|\A'|\leq t$. 
  Let $H$ be a connected component of the graph obtained from $\F$ by the removal of all edges $e$ for which in $\ET{\F}$ it holds $\lev{}{e}>k-2n$ and such that the root $r$ of $\F$ belongs to $H$. Clearly, 
  \begin{equation}\label{eq:hH_bound}
      h(\ET{H})\leq k-2n.
  \end{equation}
  On the contrary, suppose that $|\A'|>t$. 
  By Corollary \ref{cor:lem6} we know that if $\lev{}{u_i,v_i}=k-2i+2$ (i.e., when by \eqref{eq:solution} $a_i\in \A'$), then $\{r,u_i\}\in E(H)$. 
  Moreover, by the Construction $|E(S(r))| = k-2n-t$ and hence $E(S(r))\subseteq E(H)$.
  Thus $\dg{H}{r} \ge |E(S(r))|+|\A'| = k-2n-t + |\A'| > k-2n$, which in turn implies $h(\ET{H}) > k-2n$, contrary to \eqref{eq:hH_bound}.
  
  At this moment, we remark that the edges of subgraph $H$ can be simply assigned to the levels in $\{1,\ldots,k-2n\}$. We  use this fact in the second part of the proof.
    
  Now, we prove a 'hitting property', i.e., $\A'\cap \A_j\neq\emptyset$ for each $\A_j$ with $j\in\{1,\ldots,m\}$.
  By the Construction and Lemma \ref{lem:lem11}\ref{it:11b} we know that $\eta-|\A_j|+1$ elements in $R_a$ belong to $\Lev{\varphi(j)-1}{\ET{\F}}$. If $|\A_j|$ elements in $R_m$ were additionally contained in $\Lev{\varphi(j)-1}{\ET{\F}}$, then $\vert \Lev{\varphi(j)-1}{\ET{\F}}\cap R \vert > \eta$, which contradicts Lemma \ref{lem:lem11}\ref{it:11c}. Therefore at least one element in $R_m$, say the one related to $\PM$ attached at $T_j$ of $T(a_{i^*})$, must be assigned to a level different from $\varphi(j)-1$ which by Lemma \ref{lem:lem11}\ref{it:11a} is exactly the level $\varphi(j)$.
  Consequently, following Lemma \ref{lem:lem9}\ref{it:9a} the edge $\{r,u_{i^*}\}$ cannot be assigned to the level $k-2{i^*}+2$ and hence by Lemma \ref{lem:lem6} we have $\lev{}{\{u_{i^*},v_{i^*}\}}=k-2i^*+2$ which by \eqref{eq:solution} results in $a_{i^*}\in \A'$.

  ($\Leftarrow$) Now, we show that if there exists a solution $\A'$ to the $\MHS$ problem such that $|\A'|\leq t$, then there exists a $\B$-bounded elimination forest $\ET{\F}$ such that $h(\ET{\F})\le k$.
  Given $\A'$, on the edge set of $\F$ we define a function $\ell$ with maximum at most $k$. We start with the edges $\{u_i,v_i\}$ and $\{r,u_i\}$ for all $i\in\{1,\ldots,n\}$. Namely, if $a_i\in \A'$, then we set $\lev{}{\{u_i,v_i\}}=k-2i+2$ and $\lev{}{\{r,u_i\}}=k-2i+2$ otherwise. 
  As we already know the choice of the edge for the above assignment is crucial for the visibility of levels from the vertex $v_i$ (recall the two major variants depicted in Figure \ref{fig:Gbis_color}). 
  Moreover, from Lemmas \ref{lem:lem9} and \ref{lem:lem11}\ref{it:11a} we know that this decides on the possibility of assigning appropriate elements in $R_m$ to the level $\varphi(j)$ instead of $\varphi(j)-1$ which results in 'lifting' appropriate elimination subtrees, thus preserving the bound $\B$ on the level size. 
    
  Let $\ell$ be defined as in Lemmas \ref{lem:lem3a}--\ref{lem:lem10}  and as we remark in the first part of this proof (for an illustration see Figures \ref{fig:Gadg_color}--\ref{fig:Gbis_color}).
  It remains to consider $\ell$ for paths $\PM$ attached at subgraphs $T_j$ of the main component $T$. 
  By assumption, for each $j\in\{1,\ldots,m\}$ there exists $a_{i^*}\in \A'\cap \A_j$ and hence we set $\lev{}{\{u_{i^*},v_{i^*}\}}=k-2{i^*}+2$. 
  Consequently, according to Lemma \ref{lem:lem9}\ref{it:9b} the levels in $\{2,\ldots,\varphi(j)+1\}$ are admissible for the edges of $S(w_2)$ in each $T_j$ of $T(a_{i^*})$. Let $\ell(e)=\varphi(j)+1$, where $e$ is an edge of $S(w_2)$ sharing an end-vertex with the $\PM$ attached at $T_j$ of $T(a_{i^*})$. This allows $r(\PM)$ to be assigned to the level $\varphi(j)$ (with the other edges of the $\PM$ assigned to levels $\varphi(j)-M+1,\ldots,\varphi(j)-1$). For the remaining $|\A_j|-1$ paths $\PM$, their root edges we assign to the level $\varphi(j)-1$. 
  Considering the elements in $R_a$, independently, for each $j\in\{1,\ldots,m\}$ we get $\eta-|\A_j|-1$, $n-1$ and $n$ root edges (of the paths $\PM$ attached at components $\Tbard$ with $d\in\{\varphi(j)-1,\varphi(j),\varphi(j)+1\}$) that accordingly to Lemma \ref{lem:lem11}\ref{it:11b} we have already assigned the levels $\varphi(j)-1,\varphi(j)$ and $\varphi(j)+1$, respectively. 
  Thus summing the elements in $R_m$ and $R_a$ for each level $p\in\{\varphi(j)-1,\varphi(j),\varphi(j)+1\}$ we obtain $|\Lev{p}{\ET{\F}}\cap{R}|\le\eta$.
  The proof is completed by showing that in each level, $\ET{\F}$ contains at most $\B$ elements. Namely each level $p$ contains $2^{s}$ vertices corresponding to the edges of a path $\PM$ rooted at a level $p+s$ for each $s\in\{0,\ldots,\M-1\}$. Therefore 
  \begin{equation*}
   |\Lev{p}{\ET{\F}}| \le \eta\sum_{s=0}^{M-1}2^s + |W \cap \Lev{p}{\ET{\F}}| 
                     \stackrel{\eqref{eq:mbar}} {\le} n2^{\M} + \mbar
                     \stackrel{\eqref{eq:b}} {=} \B.
  \end{equation*}
\end{proof}

\begin{proof}[Proof of Theorem~\ref{thm:NPC_trees1}] 
   We give a polynomial-time reduction from the $\BTD$ problem whose $\NP$-completeness for line graphs of forests we settled in Theorem \ref{thm:NPC_forests}.
   
   \begin{figure}[htb!]
    \centering
    \includegraphics[scale=0.72]{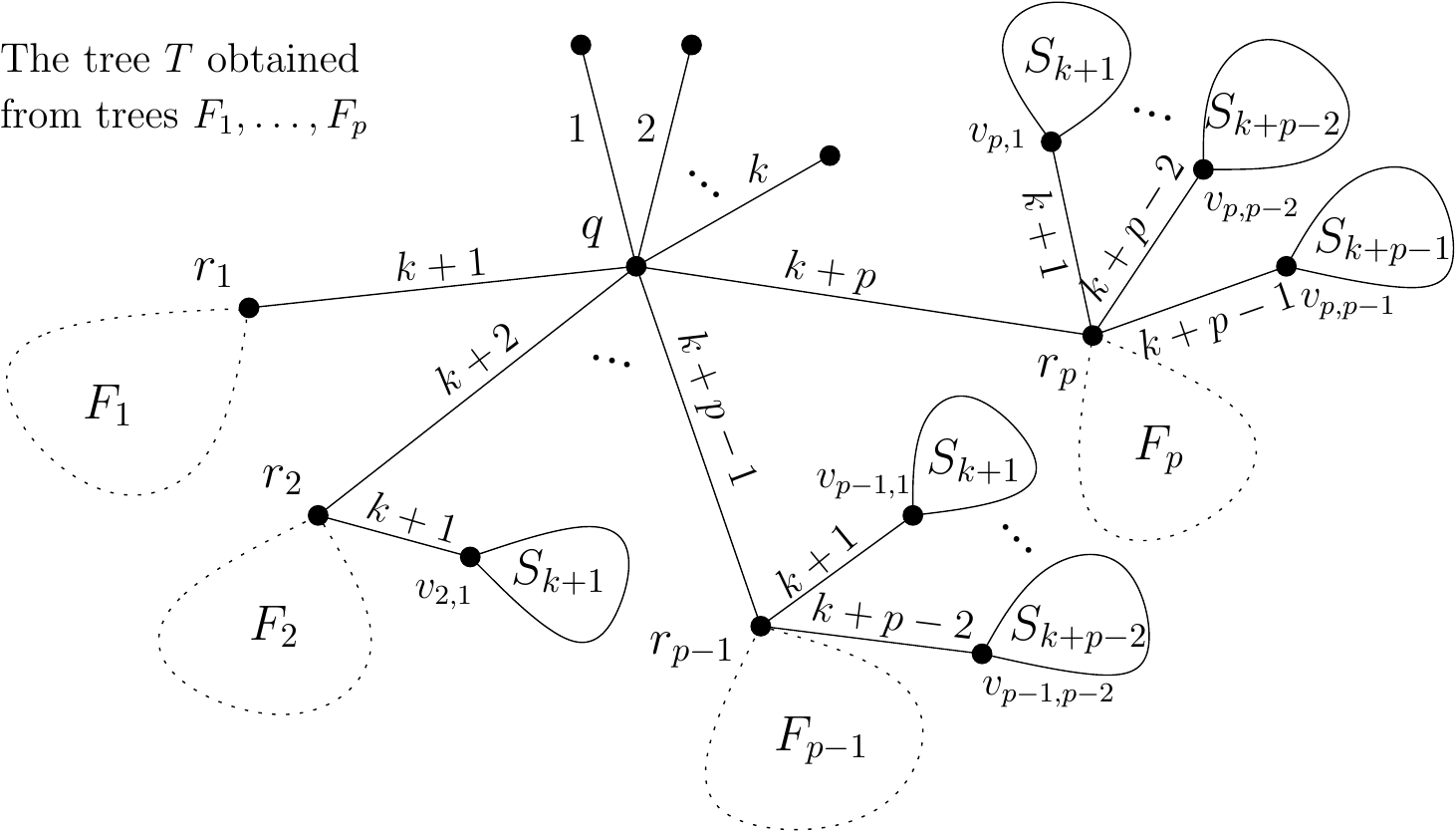}
    \caption{An illustration for the construction in the proof of Theorem~\ref{thm:NPC_trees1}.}
    \label{fig:T_NPC}
  \end{figure}
   
   Let forest $F$ and positive integers $k$, $\B$ form an instance of $\BTD$ obtained according to the construction presented in Section \ref{sec:NPC_constr}.
   Moreover, let $F_1,\ldots,F_p$ stand for the connected components of $F$ with the roots denoted by $r_1,\ldots,r_p$, respectively. 
   Given such an instance we construct a tree $T$ (see Figure \ref{fig:T_NPC}) and calculate the following related parameters:
   \begin{align*}
     \B_T &= \B + p(p-1)/2 + 1, \\ 
      k_T &= k + p.
   \end{align*}
   The tree $T$ is obtained by taking the components $F_1,\ldots,F_p$, joining their roots with a new vertex $q$ (the root of $T$) and adding $k$ new edges such that they become incident to $q$ (i.e., we get $S(q)$ being a star $S_{k+1}$), and then for each root $r_s$ with $s\in\{2,\ldots,p\}$ adding $s-1$ new stars $S_{k+2},\ldots,S_{k+s}$ and identifying a single non-central vertex of each star with the root $r_s$. The centers of the stars corresponding to the vertex $r_s$ are denoted by $v_{s,1},\ldots,v_{s,s-1}$, respectively.
   
   In order to argue that if there exists a $\B$-bounded elimination forest $\ET{F}$ of height $k$, then there exists a $\B_T$-bounded elimination tree $\ET{T}$ with $h(\ET{T})=k_T$, we define a level function $\ell$ on the vertex set of the tree $T$ such that the maximum of $\ell$ does not exceed $k_T$. 
   In fact, we also show that each edge $\{q,r_s\}$ with $s\in\{1,\ldots,p\}$ must be assigned to the level $k+s$.
   
   Consider the components $F_s$ with $s\in\{1,\ldots,p\}$.
   The level functions $\ell$ defined in the preceding lemmas can be directly used to obtain appropriate elimination trees $\ET{F_s}$ of height $k$ with the property that all levels in $\{1,\ldots,k\}$ are visible in $F_s$ from its root $r_s$. Consequently, taking into account $s-1$ stars with the centers $v_{s,1},\ldots,v_{s,s-1}$ adjacent to $r_s$, we use Lemma \ref{lem:basic} to infer that for each $s\in\{1,\ldots,p\}$ all levels in $\{1,\ldots,k+s-1\}$ are visible from $r_s$ in $T$. This in turn implies $\lev{}{\{q,r_s\}}\ge k+s$ for all $s\in\{1,\ldots,p\}$.
   Since by assumption $h(\ET{T})\le k+p$, we get $\lev{}{\{q,r_p\}}= k+p$. Now, simple induction going downwards on $s$ shows the desired equalities $\lev{}{\{q,r_s\}}= k+s$ for all $s$ in concern. The values of the level function for the remaining edges follow easily by Lemma \ref{lem:basic}. 
   
   Since by assumption $\ET{F}$ is $\B$-bounded, concerning the width of $\ET{T}$ it is enough to estimate the contribution of the edges of $T$ that do not belong to $E(F)$. Observe that the set of such edges can be partitioned into $p(p-1)/2+1$ subsets each of which constitutes the edge set of a star with the center either in $q$ or $v_{s,1},\ldots,v_{s,s-1}$, where $s\in\{2,\ldots,p\}$.
   Now, the claim follows by the fact that each edge of an arbitrary star contributes to a different level, and because each of the above-mentioned $p(p-1)/2+1$ stars must contribute to the first level.
   
   The proof that given a $\B_T$-bounded elimination tree $\ET{T}$ of height $k_T$ it is possible to find a $\B$-bounded elimination forest $\ET{F}$ of height at most $k$ is based on analogous properties.
\end{proof}

\section{The approximation algorithm} \label{sec:algorithm}

We start with a few additional definitions.
For an elimination tree $\ET{T}$ and its vertex $v$, we denote by $\ET{T}[v]$ the subtree of $\ET{T}$ induced by $v$ and all its descendants in $\ET{T}$.
By \emph{lowering} a vertex $v$ we mean a two-phase operation of moving all vertices in $\ET{T}[v]$ one level down in $\ET{T}$ (i.e., decreasing $\lev{}{u}$ by $1$ for each $u\in V(\ET{T}[v])$) and if for the resulting function there is a vertex $u\in V(\ET{T}[v])$ with $\lev{}{u}\leq 0$, then incrementing $\lev{}{w}$ for each $w\in V(\ET{T})$.
(The former `normalization' is to ensure that levels are positive integers.)
Note that this is always feasible, however it may produce an elimination tree whose height is larger than that of the initial tree.
We say that an elimination tree $\ET{T}$ is \emph{compact} if every subtree of $\ET{T}$ occupies a set of consecutive levels in $\ET{T}$.
This condition can be easily ensured in linear time. Namely, for each $\{v,u\} \in E(\ET{T})$ such that $\lev{}{v} < \lev{}{u}-1$, increment the level of $v$ and all its descendants.

Given an elimination tree $\ET{T}$, we distinguish a specific type of the root-leaf paths. Namely, by a \emph{trunk} we mean a path $(v_1,\ldots,v_k)$ from the root $v_k$ to an arbitrary leaf $v_1$ at level $1$ of $\ET{T}$ (note that only leaves at level $1$ are admissible).
We use trunks to define branches at the vertices $v_i$ with  $i\in\{2,\ldots,k\}$ of the trunk $(v_1,\ldots,v_k)$, where for particular vertex $v_i$ a \emph{branch} is understood as a subtree $\ET{T}[v]$ with $v$ being a child of $v_i$ such that $v\neq v_{i-1}$. Since elimination trees we consider are binary trees, a branch with respect to a given trunk is uniquely determined.
For any elimination tree $\ET{T}$, the $\B$ highest levels are called its \emph{prefix}.
We also say that a subtree of $\ET{T}$ is \emph{small} if it has at most $\B$ vertices; otherwise it is \emph{large}.
A subtree of $\ET{T}$ is \emph{thin} if in each level it has at most one vertex. A level of an elimination tree is \emph{full} if it has at least $\B$ vertices.

\subsection{Preprocessing steps} \label{sec:preprocessing}

In this section we introduce operations of stretching and sorting of elimination trees that constitute two major steps of the preprocessing phase of our algorithm.

We start with an operation of \emph{stretching} a subtree $\ET{T}[v]$.
If the subtree is thin or its height is at least $\B$, then stretching leaves the subtree unchanged.
Otherwise, as we will see later on, it is enough to consider the case in which a vertex $v$ at level $l$ has two children $u_1$ and $u_2$ at level $l-1$ such that both $\ET{T}[u_1]$ and $\ET{T}[u_2]$ are thin and compact.
Let $l_i$ denote the lowest level occupied by a vertex of $\ET{T}[u_i]$, $i\in\{1,2\}$, so the levels occupied by the vertices of $\ET{T}[u_i]$ are $l_i,\ldots,l-1$.
Then, stretching $\ET{T}[v]$ is realized by lowering $l-l_1$ times the vertex $u_2$ so that it is placed at the level $l_1-1$. Consequently, $\ET{T}[v]$ becomes thin and occupies at most $2l-l_1-l_2+1$ levels with the root $v$ at level $l$, the vertices of $\ET{T}[u_1]$ remaining at levels in $l_1,\ldots,l-1$ and the vertices of $\ET{T}[u_2]$ moved to the levels in $l_1-l+l_2,\ldots,l_1-1$.
(For simplicity of description, we refered here to the levels of all vertices according to the level function of the initial elimination tree, i.e., before it has been subject to the `normalization' applied after each lowering of a vertex.)
For an elimination tree $\ET{T}$, a \emph{stretched} elimination tree $\ET{T}'$ is obtained by stretching each subtree $\ET{T}[v]$ of $\ET{T}$, where the roots $v$ are selected in postorder fashion, i.e., for each  $v$ we stretch the subtrees rooted at the children of $v$ before stretching $\ET{T}[v]$.
Clearly, the transition from $\ET{T}$ to $\ET{T}'$ can be computed in linear time.
\begin{observation} \label{obs:weakly-compact} 
If $\ET{T}'$ is a stretched elimination tree obtained from a compact elimination tree $\ET{T}$, then $h(\ET{T}')\leq h(\ET{T})+2\B$ and $\ET{T}'$ is compact.
Moreover, if for a vertex $v$, $\ET{T}[v]$ is large, then $\ET{T}'[v]$ has height at least $\B$.
\end{observation}
\begin{proof}
Consider an arbitrary vertex $\ET{T}[v]$ for which stretching performs a modification.
Thus, the subtrees rooted at the children of $v$ are thin and both are of height at most $\B-2$ because the height of $\ET{T}[v]$ is at most $\B-1$.
Then, stretching $\ET{T}[v]$ gives a new subtree rooted at $v$ whose height equals the number of vertices of $\ET{T}[v]$, which is at most $1+2(\B-2)=2\B-1$.
\end{proof}

Suppose that $\ET{T}$ is an elimination tree with a vertex $z_1$ and its child $z_2$.
A \emph{switch} of $z_1$ and $z_2$ is the operation performed on $\ET{T}$ as shown in Figure~\ref{fig:switch} (intuitively, the switch presented in the figure results in exchanging the roles of the subtrees $\ET{T}[z_3]$ and $\ET{T}[z_4]$ while the aim of switching is to obtain a type of ordering of the subtrees).
We note that this concept has been used on a wider class of non-binary elimination trees under the names of tree rotations or reorderings; we refer the reader e.g. to~\cite{Liu88}.
At this point we define switching with no connection to whether $z_1$ and $z_2$ belong to a given trunk or not. Later on, we will use this operation with respect to the location of particular trunks.
It is not hard to see that if the above switch operation is performed on an elimination tree for graph $G$, then the resulting tree is also an elimination tree for that graph.
   \begin{figure*}[htb!]
    \begin{center}
    \includegraphics[scale=0.8]{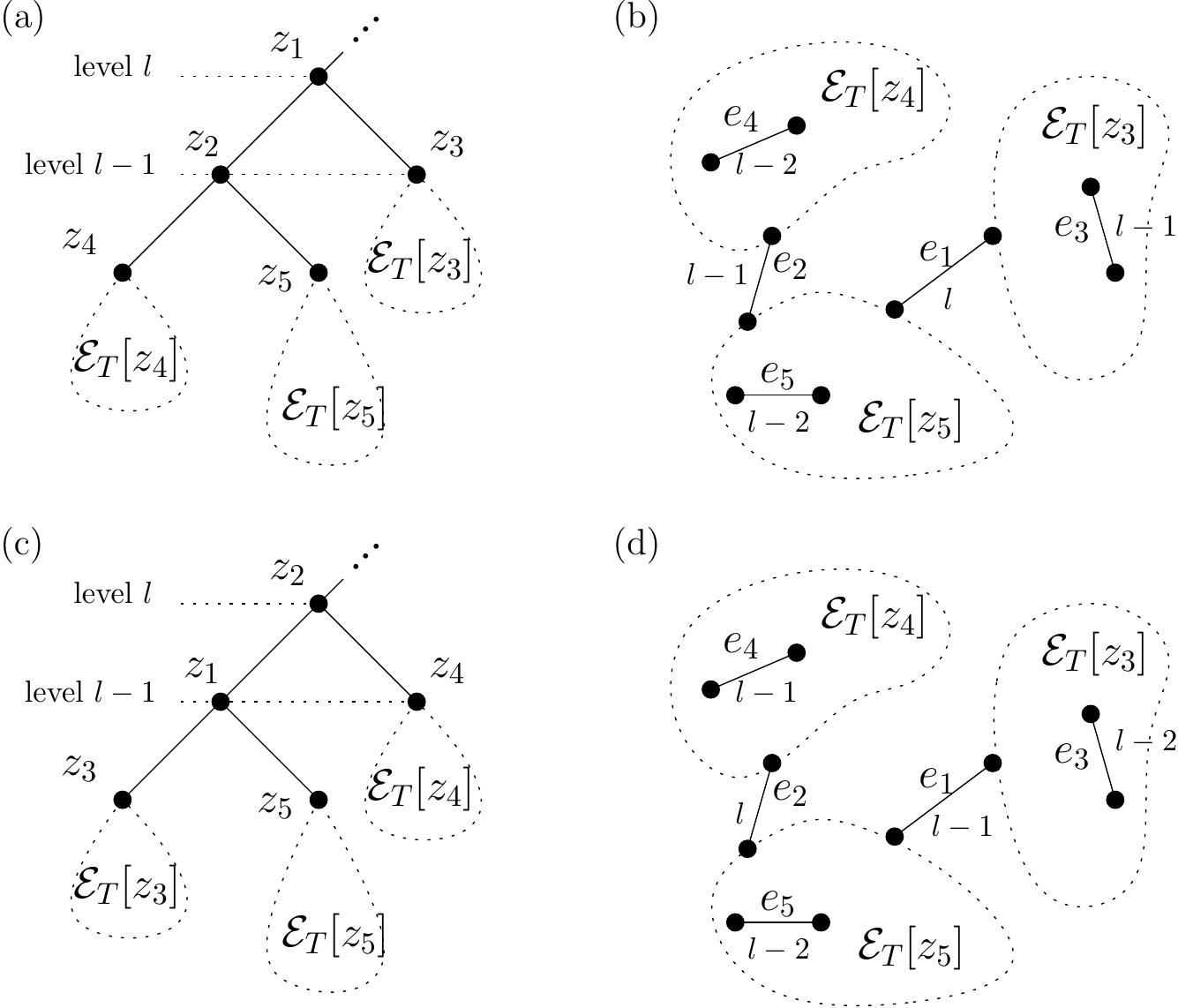}
    \caption{A switch of $z_1$ and $z_2$: (a) $\ET{T}$; (b) the levels shown on the tree $T$, where each vertex $z_i$ in $\ET{T}$ corresponds to the edge $e_i$ in $T$; (c) $\ET{T}$ after the switch; (d) the corresponding levels in $T$.}
    \label{fig:switch}
   \end{center}
   \end{figure*}

We say that $\ET{T}$ is \emph{sorted along} a trunk $(v_1,\ldots,v_k)$ (recall that $v_k$ is the root) if the height of the branch at $v_i$ is not larger than the height of the branch at $v_{i+1}$, where $i\in\{2,\ldots,k-1\}$.
Now, we say that $\ET{T}$ is \emph{sorted} if for each vertex $v$ of $\ET{T}$, the subtree $\ET{T}[v]$ is sorted along each of its trunks.
It is important to note that the notion of a trunk  pertains to any elimination tree and hence can be equally applied not only to $\ET{T}$ but also to any of its subtrees.

\begin{lemma} \label{lem:sorted-exists}
For each tree $T$, there exists a sorted $\ET{T}$ of minimum height and it can be computed in polynomial time.
\end{lemma}
\begin{proof}
For an elimination tree of minimum height, consider the number of vertices $z_1$ such that the branch at $z_1$ has smaller height than the branch at its child $z_2$.
Call it a \emph{reversed pair}.
(We use the symbols from the definition of the switch; see also Figure~\ref{fig:switch}).
For each reversed pair consider the distance $i$ of the vertex $z_1$ from the root and consider a parameter that we call the \emph{level count} that is the sum of $i$'s for all the reversed pairs.

Take the lowest vertex $z_1$ that constitutes a reversed pair with $z_2$.
The switch shown in Figure~\ref{fig:switch} results in the situation that $z_1$ and $z_2$ are no longer a reversed pair.
Note that $z_1$ may form a new reversed pair with its child.
However, the level count increases by at least one.
The final level count for any elimination tree is $O(n^2)$.
Moreover, we start with any elimination tree of minimum height, which can be computed in linear time \cite{LamY01,MozesOW08}.
Thus, sorting takes polynomial time.
\end{proof}

\subsection{Formulation of the algorithm} \label{sec:algorithm-statement}

Informally, the algorithm proceeds as follows.
It starts with a minimum height elimination tree that is stretched and sorted.
In each iteration of the main loop of the algorithm the following takes place.
If the current elimination tree is $\B$-bounded, then the algorithm finishes by returning this tree.
Otherwise, the algorithm finds the highest level that has more than $\B$ vertices and takes the lowest subtree rooted at a vertex $v$ of this level.
Then, $v$ is lowered.
\begin{center}
\begin{minipage}{.95\linewidth}
\begin{algorithm}[H]
\SetAlgoRefName{BTD}
	\caption{bounded-width tree-depth approximation for $L(T)$}
	\label{alg:treedepth}
	Let $\ET{T}$ be a minimum height compact elimination tree for $L(T)$\;
	Obtain $\ET{T}^0$ by first making $\ET{T}$ stretched and then sorted\;
	$t\gets0$\;
	\While{$\ET{T}^t$ is not $\B$-bounded}
	{
		Find the highest level $l_t$ such that $\vert\Lev{l_t}{\ET{T}^t}\vert>\B$\;
		$v\gets\arg\min_u h(\ET{T}^t[u])$, where $u$ iterates over all vertices in $\Lev{l_t}{\ET{T}^t}$\;
		Obtain $\ET{T}^{t+1}$ by lowering $v$ in $\ET{T}^t$ and increment $t$\;
	}
	\Return $\ET{T}^t$
\end{algorithm}
\end{minipage}
\end{center}

\subsection{Analysis} \label{sec:analysis}

Since lowering any vertex results in an elimination tree, the final tree  $\ET{T}^{\tau}$ is an elimination tree.
In each iteration of the main loop the width of the currently highest level $l_t$ of size larger than $\B$ decreases.
Thus, the number of iterations can be trivially bounded by $O(n^2)$.
Hence, $\ET{T}^{\tau}$ is a $\B$-bounded elimination tree for $L(T)$.
It remains to upper-bound its height.
An outline of the argument is as follows.

First, we introduce a structural property of an elimination tree: if we iterate over the consecutive levels, starting at the level at distance $\B$ from the root and going `downwards', then we keep processing full levels for a number of steps and once we reach a level $l$ that is not full, all the remaining levels below $l$ are not full as well.
More formally, let $l_1<l_2<\cdots<l_d$ be the full levels of an elimination tree $\ET{T}$, except for the highest $\B-1$ levels.
We say that $\ET{T}$ is \emph{structured} if either all full levels belong to the prefix or $l_d=h(\ET{T})-\B+1$ and $l_i=l_{i-1}+1$ for each $i\in\{2,\ldots,d\}$.
In other words, in a structured elimination tree the full levels that are not in the prefix form a consecutive segment that reaches the lowest level of the prefix minus one.
In this context we prove that $\ET{T}^0$ obtained at the beginning of Algorithm~\ref{alg:treedepth} has the above property (cf. Lemma~\ref{lem:Tzero-structured}).

Next, it is important that we keep another invariant along the iterations of the main loop (i.e., when we transform $\ET{T}^{t-1}$ to $\ET{T}^{t}$).
Namely, if some level $h(\ET{T}^{t})-l$ \emph{becomes} full, then the level right above it (i.e., $h(\ET{T}^{t})-l+1$) should be also full at this point, and both levels should stay full in the future iterations: the levels $h(\ET{T}^{t'})-l$ and $h(\ET{T}^{t'})-l+1$ are full for each $t'\geq t$.
This is proved in Lemma~\ref{lem:correct-filling} and essentially means that starting with a structured $\ET{T}^0$, subsequent trees $\ET{T}^t$, $t>0$, remain structured.

Finally, in Lemma~\ref{lem:almost-all-full} we prove that in the final elimination tree $\ET{T}^{\tau}$, if its height increased with respect to $\ET{T}^0$, then at most $\B$ levels are not full.
Moreover, in the case when $h(\ET{T}^{\tau})=h(\ET{T}^0)$, due to Observation~\ref{obs:weakly-compact} the height is additively at most $2\B$ from $\btd{L(T)}{\B}$.
Altogether, the above properties ensure the required approximation bound.
\begin{lemma} \label{lem:Tzero-structured}
$\ET{T}^0$ is structured.
\end{lemma}
\begin{proof}
We may assume that $\ET{T}^0$ is large as otherwise the lemma follows.
Note that $\ET{T}^0$ is compact due to Observation~\ref{obs:weakly-compact} and the fact that $\ET{T}$ computed at the beginning of Algorithm~\ref{alg:treedepth} is compact.
We describe the following process of constructing a subset $D$ of vertices of $\ET{T}^0$. 
Let initially $D$ consist of the root only.
Suppose that there exists $v\in D$ such that $\ET{T}^0[v]$ is not thin and $v$ has two children $u$ and $u'$ such that both $\ET{T}^0[u]$ and $\ET{T}^0[u']$ have vertices at the lowest level of the prefix.
Then do the following: remove $v$ from $D$ and add $u$ and $u'$ to $D$; call it a \emph{split of} $v$.
Repeat this modification splitting the vertices in $D$ as long as $D$ contains a vertex $v$ that can be split. Note that no vertex $v$ at the lowest level of the prefix admits splitting, since $v$ has no child in the prefix and hence the process of subsequent splits will not continue for vertices below the level $h(\ET{T}^0) - b + 2$. 
Observe also that the split of $v$ implies that $u$ and $u'$ are one level lower than $v$ because $\ET{T}^0[v]$ is not thin.
Hence, by construction, $v$ has two children that have not been lowered during stretching.
For any two vertices in $D$, they are not related in $\ET{T}^0$, i.e., one is not the ancestor/descendant of the other.

As the first case consider $|D|\geq\B$.
This implies that the lowest level of the prefix is full because each vertex in $D$ provides a unique vertex that belongs to the lowest level $h(\ET{T}^0)-\B+1$ of the prefix.
Let $l\leq h(\ET{T}^0)-\B+1$ be the lowest level such that all levels between $l$ and $h(\ET{T}^0)-\B+1$ are full.
Thus, in order to prove that $\ET{T}^0$ is structured it is enough to show that no level below $l$ is full.
We will say that a subtree $\ET{T}^0[u]$, for any vertex $u$, is \emph{simple} if $\ET{T}^0[u]$ has a trunk such that each of its branches is thin.
(Informally, such a subtree is a leaf-root path with thin subtrees attached to the vertices of the path.)
Consider any $v\in\Lev{l}{\ET{T}^0}$ at level $l$.
We will argue by contradiction that $\ET{T}^0[v]$ is simple.
If $\ET{T}^0[v]$ is not simple, then both children of $v$ are roots of large subtrees.
But since $\ET{T}^0$ is sorted, each ancestor of $v$ has two children that are large.
Since $l\leq h(\ET{T}^0)-\B+1$, i.e., $v$ does not belong to the prefix, $v$ has at least $\B-1$ ancestors.
This however implies that, considering $v$ and its $\B-1$ immediate ancestors, at least $\B$ pairwise different subtrees intersect level $l-1$.
Thus, level $l-1$ is full --- a contradiction with the choice of $l$.
Then, using the fact that $\ET{T}^0[v]$ is simple, we have that for each $v\in \Lev{l}{\ET{T}^0}$ there exists its farthest ancestor $a(v)$ such that $\ET{T}^0[a(v)]$ is simple (note that $a(v)=v$ is allowed thus $a(v)$ exists).
Moreover, the vertices $a(v)$ are pairwise different by the definition of a simple tree.
For each $v\in \Lev{l}{\ET{T}^0}$, since $\ET{T}^0[a(v)]$ is sorted,
\begin{equation} \label{eq:levels-decrease}
|V(\ET{T}^0[a(v)])\cap\Lev{l'-1}{\ET{T}^0}|\leq|V(\ET{T}^0[a(v)])\cap\Lev{l'}{\ET{T}^0}|
\end{equation}
for each level $l'\leq l$.
(Intuitively and informally, as we go down along the tree, the sizes of the levels decrease.)
%This holds because $\ET{T}^0[a(v)]$ has a branch that does not intersect level $l$ for otherwise $\B-1$ immediate ancestors of $v$ have branches that intersect with level $l-1$ proving that level $l-1$ is full, which is not possible due to the choice of $l$.
Using this inequality for each individual simple tree $\ET{T}^0[a(v)]$ for $v\in \Lev{l}{\ET{T}^0}$ and taking a union of all of them, we have
$|\Lev{l-1}{\ET{T}^0}|\geq|\Lev{l-2}{\ET{T}^0}|\geq\cdots\geq |\Lev{l'}{\ET{T}^0}|$ for each $l'\leq l-2$.
Using $|\Lev{l-1}{\ET{T}^0}|<\B$ (level $l-1$ is not full because of the choice of $l$), we have that $|\Lev{l'}{\ET{T}^0}|<\B$ for each $l'<l$, which proves that $\ET{T}^0$ is structured in case when $|\Lev{l}{\ET{T}^0}|\geq\B$.

Suppose now that $|D|<\B$. First observe that no vertex in $D$ belongs to the lowest level of the prefix because otherwise a split of each ancestor of $v$ has occurred, which would mean that $|D|\geq\B$.
The latter is true because each split by definition increments the size of $D$.
For each $v\in D$, the subtree $\ET{T}^0[v]$ is simple, since otherwise $v$ has two children whose subtrees are large.
By construction, these children are located at level directly below $v$.
Since the subtrees are large, they intersect the lowest level of the prefix.
Moreover, we argued that $v$ is not at the lowest level of the prefix and hence $v$ is admissible for a split, which gives a contradiction.
Since $\ET{T}^0[v]$ is sorted (because $\ET{T}^0$ is sorted), the lengths of the branches do not increase as we go from the root of $\ET{T}^0[v]$ towards the leaf of the trunk.
This means that for each level $l\leq h(\ET{T}^0)-\B+1$, 
\[|V(\ET{T}^0[v])\cap\Lev{l}{\ET{T}^0}|\geq|V(\ET{T}^0[v])\cap\Lev{l-1}{\ET{T}^0}|.\]
This inequality crucially depends on the fact that the branch at $v$ (which is the longest branch) does not intersect the lowest level of the prefix, and on the fact that no vertex in $D$ is at the lowest level of the prefix.
(Note that the branch at $v$ does not intersect the lowest level of the prefix for otherwise there would be a split at $v$.)
Thus, the inequality implies that if the lowest level of the prefix has less than $\B$ vertices, then each level that is not in the prefix has less than $\B$ vertices, which is argued similarly as on the basis of Inequality~\eqref{eq:levels-decrease} in case of $|D|\geq\B$, which completes the proof.
\end{proof}

\begin{lemma} \label{lem:correct-filling}
Consider an arbitrary $t\geq0$ and the sizes $s_{t-1}$ and $s_{t}$ of the levels at the same fixed distance $d'$ from the roots in $\ET{T}^{t-1}$ and $\ET{T}^{t}$, respectively.
If $s_{t-1}\geq\B$, then $s_{t}\geq \B$.
If $s_{t-1}<\B$ and $s_{t}\geq \B$, then the level at distance $d'-1$ from the root in $\ET{T}^{t}$ is full. 
\end{lemma}
\begin{proof}
Denote by $l_{t-1}=h(\ET{T}^{t-1})-d'$ and $l_{t}=h(\ET{T}^{t})-d'$ the two levels of $\ET{T}^{t-1}$ and $\ET{T}^{t}$ at distance $d'$ from the root, respectively.
Let $d$ be the level at which there is the root of the subtree of $\ET{T}^{t-1}$ that is lowered in the $t$-th iteration.
Denote by $v_1,\ldots,v_p$ the vertices of the level $d$ in $\ET{T}^{t-1}$ sorted so that $h(\ET{T}^{t-1}[v_1])\geq\cdots\geq h(\ET{T}^{t-1}[v_p])$.
In the $t$-th iteration, $p=|\Lev{d}{\ET{T}^{t-1}}|>\B$.
The lowered vertex is $v_p$ because the algorithm selects the  subtree of minimum height.

If $s_{t-1}=s_t$, then the lemma follows so assume that $s_{t-1}\neq s_{t}$. Thus, it holds $d\geq l_{t-1}$ because if level $d$ is below the level $l_{t-1}$, then lowering the subtree $\ET{T}^{t-1}[v_p]$ with $v_p$ at level $d$ would not change the sizes of levels above level $d$, in particular the size of level $l_{t-1}$, which would imply $s_{t-1}=s_t$ and contradict our assumption.
By the same argument, the subtree $\ET{T}^{t-1}[v_p]$ intersects the level $l_{t}$ in $\ET{T}^{t}$.

Since $\ET{T}^0$ is compact, the subtree $\ET{T}^{t-1}[v_p]$ is compact.
If $\ET{T}^{t-1}[v_p]$ is small, then we use the following argument: $\ET{T}^{t-1}[v_p]$ is thin and thus in particular has a one-vertex intersection with the level $l_{t}$ in $\ET{T}^{t}$ and, since $s_{t-1}\neq s_{t}$, it does not intersect the level $l_{t-1}$ in $\ET{T}^{t-1}$.
Since it is compact, it intersects all levels $l_{t-1}+1,\ldots,d$ in $\ET{T}^{t-1}$.
Moreover, $s_{t}=s_{t-1}+1$, which completes the proof.
Also, if $s_{t}\geq\B$, then $p$ subtrees intersect the level $l_{t}-1$ in $\ET{T}^{t}$, which proves the second claim from the lemma because $p>\B$.

If all $p$ subtrees rooted at the level $d$ in $\ET{T}^{t-1}$ are large, then due to the ordering of them and since $\ET{T}^{0}$ is compact, all of these subtrees intersect the level $l_{t}$.
Hence $s_t\geq \B$.
\end{proof}

\begin{lemma} \label{lem:almost-all-full}
Let $\ET{T}^{\tau}$ be the elimination tree returned by Algorithm~\ref{alg:treedepth}.
If $h(\ET{T}^{\tau})>h(\ET{T}^{0})$, then at most $\B$ levels of $\ET{T}^{\tau}$ are not full.
\end{lemma}
\begin{proof}
The levels that may not be full in $\ET{T}^{\tau}$ are $\B-1$ highest levels and level $1$.
The argument follows by an induction on the number of iterations.
The base case is covered by Lemma~\ref{lem:Tzero-structured}, and the inductive step is due to Lemma~\ref{lem:correct-filling}.
More precisely, if the height of the elimination tree increases in a particular iteration, then level $2$ must be full, and consequently all the other levels except for the highest $\B-1$ and level $1$ are full.
\end{proof}

\begin{proof}[Proof of Theorem~\ref{thm:ALG}]
The complexity follows from Lemma~\ref{lem:sorted-exists}.

Suppose first that $h(\ET{T}^{\tau})=h(\ET{T}^{0})$.
By Observation~\ref{obs:weakly-compact}, $h(\ET{T}^0)\leq h(\ET{T})+2\B$ because sorting does not increase the height of an elimination tree.
Hence in this case the theorem follows because $\btd{L(T)}{\B}\geq h(\ET{T}^0)$.
Suppose now that $h(\ET{T}^{\tau})>h(\ET{T}^{0})$.
By Lemma~\ref{lem:almost-all-full}, at most $\B$ levels of $\ET{T}^{\tau}$ are not full.
Thus, $h(\ET{T}^{\tau})\leq \B+\frac{m}{\B}$, where $m$ is the number of edges of $T$.
The lower bound of $\btd{L(T)}{\B}\geq\frac{m}{\B}$ completes the proof.
\end{proof}

\bibliographystyle{plain}
\bibliography{bibliography}

\begin{thebibliography}{10}

\bibitem{Ben-AsherFN99}
Yosi Ben{-}Asher, Eitan Farchi, and Ilan Newman.
\newblock Optimal search in trees.
\newblock {\em {SIAM} J. Comput.}, 28(6):2090--2102, 1999.

\bibitem{BodlaenderDJKKMT98}
Hans~L. Bodlaender, Jitender~S. Deogun, Klaus Jansen, Ton Kloks, Dieter
  Kratsch, Haiko M{\"{u}}ller, and Zsolt Tuza.
\newblock Rankings of graphs.
\newblock {\em {SIAM} J. Discrete Math.}, 11(1):168--181, 1998.

\bibitem{bode13}
Piotr Borowiecki and Dariusz Dereniowski.
\newblock On-line ranking of split graphs.
\newblock {\em Discrete Mathematics {\&} Theoretical Computer Science},
  15(2):195--214, 2013.

\bibitem{CarmoDKL04}
Renato Carmo, Jair Donadelli, Yoshiharu Kohayakawa, and Eduardo~Sany Laber.
\newblock Searching in random partially ordered sets.
\newblock {\em Theor. Comput. Sci.}, 321(1):41--57, 2004.

\bibitem{CicaleseJLV12}
Ferdinando Cicalese, Tobias Jacobs, Eduardo~Sany Laber, and Caio~Dias Valentim.
\newblock The binary identification problem for weighted trees.
\newblock {\em Theor. Comput. Sci.}, 459:100--112, 2012.

\bibitem{CicaleseKLPV16}
Ferdinando Cicalese, Bal{\'{a}}zs Keszegh, Bernard Lidick{\'{y}},
  D{\"{o}}m{\"{o}}t{\"{o}}r P{\'{a}}lv{\"{o}}lgyi, and Tom{\'{a}}\v{s} Valla.
\newblock On the tree search problem with non-uniform costs.
\newblock {\em Theor. Comput. Sci.}, 647:22--32, 2016.

\bibitem{TorreGS95}
Pilar de~la Torre, Raymond Greenlaw, and Alejandro~A. Sch{\"{a}}ffer.
\newblock Optimal edge ranking of trees in polynomial time.
\newblock {\em Algorithmica}, 13(6):592--618, 1995.

\bibitem{Dereniowski06}
Dariusz Dereniowski.
\newblock Edge ranking of weighted trees.
\newblock {\em Discrete Applied Mathematics}, 154(8):1198--1209, 2006.

\bibitem{Dereniowski08}
Dariusz Dereniowski.
\newblock Edge ranking and searching in partial orders.
\newblock {\em Discrete Applied Mathematics}, 156(13):2493--2500, 2008.

\bibitem{DereniowskiKUZ17}
Dariusz Dereniowski, Adrian Kosowski, Przemys\l{}aw Uzna\'{n}ski, and Mengchuan
  Zou.
\newblock Approximation strategies for generalized binary search in weighted
  trees.
\newblock In {\em {ICALP} 2017}, pages 84:1--84:14.

\bibitem{DereniowskiN06}
Dariusz Dereniowski and Adam Nadolski.
\newblock Vertex rankings of chordal graphs and weighted trees.
\newblock {\em Inf. Process. Lett.}, 98(3):96--100, 2006.

\bibitem{Emamjomeh-Zadeh17}
Ehsan Emamjomeh{-}Zadeh and David Kempe.
\newblock A general framework for robust interactive learning.
\newblock In {\em {NIPS} 2017}, pages 7082--7091.

\bibitem{Emamjomeh-Zadeh16}
Ehsan Emamjomeh{-}Zadeh, David Kempe, and Vikrant Singhal.
\newblock Deterministic and probabilistic binary search in graphs.
\newblock In {\em {STOC} 2016}, pages 519--532.

\bibitem{GiannopoulouHT12}
Archontia~C. Giannopoulou, Paul Hunter, and Dimitrios~M. Thilikos.
\newblock Lifo-search: {A} min-max theorem and a searching game for cycle-rank
  and tree-depth.
\newblock {\em Discrete Applied Mathematics}, 160(15):2089--2097, 2012.

\bibitem{IyerRV91}
Ananth~V. Iyer, H.~Donald Ratliff, and Gopalakrishnan Vijayan.
\newblock On an edge ranking problem of trees and graphs.
\newblock {\em Discret. Appl. Math.}, 30(1):43--52, 1991.

\bibitem{KarpasNS15}
Ilan Karpas, Ofer Neiman, and Shakhar Smorodinsky.
\newblock On vertex rankings of graphs and its relatives.
\newblock {\em Discrete Mathematics}, 338(8):1460--1467, 2015.

\bibitem{KatchMcCauSeg95}
M.~Katchalski, W.~McCaugh, and S.~Seager.
\newblock Ordered colourings.
\newblock {\em Discrete Mathematics}, 142:141--154, 1995.

\bibitem{LaberMP02}
Eduardo~Sany Laber, Ruy~Luiz Milidi{\'{u}}, and Artur~Alves Pessoa.
\newblock On binary searching with nonuniform costs.
\newblock {\em {SIAM} J. Comput.}, 31(4):1022--1047, 2002.

\bibitem{LamY98}
Tak~Wah Lam and Fung~Ling Yue.
\newblock Edge ranking of graphs is hard.
\newblock {\em Discrete Applied Mathematics}, 85(1):71--86, 1998.

\bibitem{LamY01}
Tak~Wah Lam and Fung~Ling Yue.
\newblock Optimal edge ranking of trees in linear time.
\newblock {\em Algorithmica}, 30(1):12--33, 2001.

\bibitem{Liu90}
Joseph~W.H. Liu.
\newblock The role of elimination trees in sparse factorization.
\newblock {\em SIAM J. Matrix Anal. Appl.}, 11(1):134--172, 1990.

\bibitem{Liu88}
J.W.H. Liu.
\newblock Equivalent sparse matrix reorderings by elimination tree rotations.
\newblock {\em SIAM J. Sci. Stat. Comput.}, 9(3):424--444, 1988.

\bibitem{MakinoUI01}
Kazuhisa Makino, Yushi Uno, and Toshihide Ibaraki.
\newblock On minimum edge ranking spanning trees.
\newblock {\em J. Algorithms}, 38(2):411--437, 2001.

\bibitem{McDonald15}
Daniel~C. McDonald.
\newblock On-line vertex ranking of trees.
\newblock {\em {SIAM} J. Discrete Math.}, 29(1):145--156, 2015.

\bibitem{MozesOW08}
Shay Mozes, Krzysztof Onak, and Oren Weimann.
\newblock Finding an optimal tree searching strategy in linear time.
\newblock In {\em {SODA} 2008}, pages 1096--1105.

\bibitem{NesetrilM06}
Jaroslav Ne\v{s}et\v{r}il and Patrice~Ossona de~Mendez.
\newblock Tree-depth, subgraph coloring and homomorphism bounds.
\newblock {\em Eur. J. Comb.}, 27(6):1022--1041, 2006.

\bibitem{OnakP06}
Krzysztof Onak and Pawe\l{} Parys.
\newblock Generalization of binary search: Searching in trees and forest-like
  partial orders.
\newblock In {\em {FOCS} 2006}, pages 379--388.

\bibitem{Schaffer89}
Alejandro~A. Sch{\"{a}}ffer.
\newblock Optimal node ranking of trees in linear time.
\newblock {\em Inf. Process. Lett.}, 33(2):91--96, 1989.

\bibitem{Zwaan10}
Ruben van~der Zwaan.
\newblock Vertex ranking with capacity.
\newblock In {\em {SOFSEM} 2010}, pages 767--778.

\bibitem{ZhouKN96}
Xiao Zhou, Abul Kashem, and Takao Nishizeki.
\newblock Generalized edge-ranking of trees (extended abstract).
\newblock In {\em {WG} 1996}, pages 390--404.

\bibitem{ZhouN94}
Xiao Zhou and Takao Nishizeki.
\newblock An efficient algorithm for edge-ranking trees.
\newblock In {\em {ESA} 1994}, pages 118--129.

\bibitem{ZhouN95}
Xiao Zhou and Takao Nishizeki.
\newblock Finding optimal edge-rankings of trees.
\newblock In {\em {SODA} 1995}, pages 122--131.

\end{thebibliography}

\end{document}